\documentclass[letterpaper, 10 pt, conference]{ieeeconf}  

\IEEEoverridecommandlockouts                              

\usepackage{authblk}

\usepackage{amsmath}
\usepackage{amssymb}
\usepackage{mathtools}
\usepackage{bbm,bm} 
\usepackage{amsfonts}
\usepackage{enumerate}
\usepackage{upgreek,esint}
\usepackage[ruled,noline,linesnumbered]{algorithm2e}
\let\oldnl\nl
\newcommand{\nonl}{\renewcommand{\nl}{\let\nl\oldnl}}
\SetInd{0.5em}{0.5em} 
\usepackage[normalem]{ulem}

\usepackage{accents}

\usepackage{xcolor}

\definecolor{tealblue}{HTML}{008080}
\definecolor{MyBlue}{HTML}{000090}

\usepackage[textwidth=15mm]{todonotes}		

\newcommand{\reviseagain}[1]{#1}
\newcommand{\old}[1]{}
\def\endedit{\color{black}}		

\definecolor{Cloud}{HTML}{F5F5F5}

\definecolor{MScolor}{HTML}{BBDDFF}

\makeatletter

\newcommand{\routinenl}{%
\nonl
  \refstepcounter{AlgoLine}%
  \nlset{\arabic{AlgoLine}}
  }
\makeatother

\usepackage[scaled]{inconsolata}

\usepackage{newfloat}
\usepackage{caption}
\usepackage{subcaption}

\DeclareFloatingEnvironment[fileext=frm,placement={!ht},name=Algorithm,within=none]{algorithmenv}  
\makeatletter
\renewcommand{\@algocf@capt@boxed}{above}
\makeatother
\makeatletter
\newcommand{\removelatexerror}{\let\@latex@error\@gobble}
\makeatother

  \makeatletter
  \let\NAT@parse\undefined
  \makeatother

\usepackage{hyperref}
\hypersetup{
colorlinks=true,
linktocpage=true,
pdfstartview=FitH,
breaklinks=true,
pdfpagemode=UseNone,
pageanchor=true,
pdfpagemode=UseOutlines,
plainpages=false,
bookmarksnumbered,
bookmarksopen=false,
bookmarksopenlevel=1,
hypertexnames=true,
urlcolor=white!30!black,linkcolor=white!30!black,citecolor=white!30!black, 
pdftitle={A privacy-preserving distributed computational approach for distribution locational marginal prices},
pdfauthor={},
pdfsubject={},
pdfkeywords={},
pdfcreator={pdfLaTeX},
pdfproducer={LaTeX with hyperref}
}

\usetikzlibrary{arrows, positioning}
\usepackage{cleveref}
\usepackage{pbox}
\usepackage{mathtools}

\usepackage{acronym}		
\newcommand{\acdef}[1]{\textit{\acl{#1}} \textup{(\acs{#1})}\acused{#1}}		
\newcommand{\acdefp}[1]{\emph{\aclp{#1}} \textup(\acsp{#1}\textup)\acused{#1}}	
\newcommand{\acndefp}[1]{\aclp{#1} \textup(\acsp{#1}\textup)\acused{#1}}	

\makeatletter
\newcommand{\pushright}[1]{\ifmeasuring@#1\else\omit\hfill$\displaystyle#1$\fi\ignorespaces}
\newcommand{\pushleft}[1]{\ifmeasuring@#1\else\omit$\displaystyle#1$\hfill\fi\ignorespaces}
\makeatother

\newtheorem{theorem}{Theorem}

\newtheorem{lemma}{Lemma}

\newtheorem{remark}{Remark}

\newcommand{\obsolete}[1]{}
\newcommand{\nocolsep}{\arraycolsep=1.4pt\def\arraystretch{1}}
\makeatletter
\newcommand{\refereq}[2]{\overset{\text{\tiny{#1}}}{#2}}\newcommand{\vast}{\bBigg@{4}}
\newcommand{\Vast}{\bBigg@{5}}
\makeatother
\newcommand{\Real}{\mathbb{R}}

\newcommand{\grad}{\nabla}

\newcommand{\inner}[2]{\langle{#1},{#2}\rangle}

\newcommand{\domain}[1]{\textup{dom}(#1)}

\newcommand{\expectationfunction}{\mathbb{E}}
\newcommand{\expectation}[1]{\expectationfunction[#1]}

\newcommand{\condexpectation}[2]{\expectationfunction[#2\cond{#1}]}

\newcommand{\probafunction}{\mathbb{P}}
\newcommand{\proba}[1]{\probafunction(#1)}

\newcommand{\filtrationfunction}{\mathcal{F}}
\newcommand{\filtrationk}[1]{\filtrationfunction_{#1}}

\usepackage{scalerel}
\newlength\bshft 
\def\fakebold#1{\ThisStyle{\ooalign{$\SavedStyle#1$\cr%
  \kern-\bshft$\SavedStyle#1$\cr%
  \kern\bshft$\SavedStyle#1$}}}

\newcommand{\zero}[1]{o(#1)}

\newcommand{\magnitude}[1]{O(#1)}

\newcommand{\cond}{|}

\newcommand{\indicatorfunction}[1]{I_{#1}}

\newcommand{\sigmaalgebra}{\sigma}
\newcommand{\sigmaalgebrax}[1]{\sigmaalgebra(#1)}

\newcommand{\distfunction}{\textup{dist}}
\newcommand{\disttfunction}[1]{\distfunction_{#1}}

\newcommand{\bigdistt}[3]{\disttfunction{#1}\big(#2,#3\big)}
\DeclareMathOperator{\diag}{diag}
\DeclareMathOperator{\minimize}{minimize}
\DeclareMathOperator{\subjectto}{subject\ to}

\DeclarePairedDelimiter{\cardinality}{|}{|}
\DeclarePairedDelimiter{\norm}{\|}{\|}
\newcommand{\normx}[2]{\norm{#1}_{#2}}
\newcommand{\Normx}[2]{\norm*{#1}_{#2}}

\newcommand{\transpose}{\top}

\newcommand{\defeq}{:=}

\newcommand{\identity}{\textup{Id}}

\DeclareMathOperator{\produced}{p}
\DeclareMathOperator{\consumed}{c}

\newcommand{\iter}{k} 
\newcommand{\altiter}{l} 
\newcommand{\dimension}{d}
\newcommand{\primaldimension}{m}
\newcommand{\primaldimensioni}[1]{\primaldimension_{#1}}
\newcommand{\Dimension}{P}
\newcommand{\dualdimension}{q}
\newcommand{\dimensioni}[1]{\primaldimension_{#1}}

\newcommand{\DimensionI}[1]{\Dimension_{#1}}
\renewcommand{\probafunction}{\textup{Prob}}
\newcommand{\Probas}{\Pi}
\newcommand{\ProbaI}[1]{\Probas_{#1}}
\newcommand{\probas}{\pi}
\newcommand{\probai}[1]{\probas_{#1}}
\newcommand{\probaij}[2]{\probas_{#1,#2}}

\newcommand{\stepsize}{\tau}

\newcommand{\Stepsize}{T}

\newcommand{\Stepsizek}[1]{\Stepsize}
\newcommand{\Stepsizeki}[2]{\Stepsizek{#1}_{#2}}
\newcommand{\StepsizekI}[2]{\Stepsizek{#1}_{#2}}
\newcommand{\StepsizeDSOk}[1]{\Stepsizek{#1}_{0}}
\newcommand{\StepsizeLAk}[2]{\Stepsizek{#2}_{#1}}

\newcommand{\StepsizeI}[1]{\Stepsize_{#1}}
\newcommand{\stepsizei}[1]{\stepsize_{#1}}

\newcommand{\dualstepsize}{\sigma}
\newcommand{\constraintcovariance}{\Sigma}

\newcommand{\costcovariance}{\Omega}

\newcommand{\Smoothness}{\Lambda}

\newcommand{\SmoothnessI}[1]{\Smoothness_{#1}}

\newcommand{\Scale}{\Delta}
\newcommand{\Scalek}[1]{\Scale}

\newcommand{\gammacoef}{\Gamma}
\newcommand{\gammacoefkl}[2]{\gammacoef_{#1}^{#2}}

  \newcommand{\stronglyconvexfunction}{\zeta}
  \newcommand{\stronglyconvex}[1]{\stronglyconvexfunction(#1)}

  \newcommand{\lyapunovfunction}{\mathcal{L}}%
\newcommand{\lyapunovk}[1]{\lyapunovfunction_{#1}}%

\newcommand{\Aggregator}{\mathcal{A}}
\newcommand{\aggregator}{a}
\newcommand{\altaggregator}{\aggregator^\prime}
\newcommand{\nbaggregators}{p}
\newcommand{\Block}{\mathcal{I}}
\newcommand{\block}{I}
\newcommand{\coordinate}{i}
\newcommand{\altcoordinate}{j}

\newcommand{\minusblock}{-\block}
\newcommand{\tim}{t}
\newcommand{\Tim}{\mathcal{T}}
\newcommand{\timlast}{T}
\newcommand{\Node}{\mathcal{N}}
\newcommand{\Nodea}[1]{\Node_{#1}}
\newcommand{\Nodep}{\Node_+}
\newcommand{\node}{n}
\newcommand{\nodenb}{N}
\newcommand{\nodem}{\node_{-}}
\newcommand{\altnode}{m}
\newcommand{\altnodem}{m_{-}}
\newcommand{\Aall}{A}
\newcommand{\ADSO}{\Aall_0}

\newcommand{\ALA}[1]{\Aall_{#1}}

\newcommand{\AI}[1]{\Aall_{#1}}
\newcommand{\bvector}{b}\newcommand{\bvectorDSO}{\bvector_0}
\newcommand{\bvectorLA}[1]{\bvector_{#1}}
\newcommand{\extcostfunction}{\Phi}
\newcommand{\extcostIfunction}[1]{\extcostfunction_{#1}}
\newcommand{\extcostI}[2]{\extcostIfunction{#1}(#2)}
\newcommand{\vextcostfunction}{\bm{\extcostfunction}}
\newcommand{\hatextcost}{\hat{\extcostfunction}}
\newcommand{\hatvextcost}{\hat{\bm{\extcostfunction}}}
\newcommand{\vextcost}[1]{\vextcostfunction(#1)}
\newcommand{\hatextcostk}[1]{\hatextcost_{#1}}
\newcommand{\hatvextcostk}[1]{\hatvextcost_{#1}}
\newcommand{\costfunction}{\phi}

\newcommand{\costDSOfunction}{\costfunction_0}

\newcommand{\costLAfunction}[1]{\costfunction_{#1}}
\newcommand{\costIfunction}[1]{\costfunction_{#1}}

\newcommand{\cost}[1]{\costfunction(#1)}
\newcommand{\extcost}[1]{\extcostfunction(#1)}
\newcommand{\costDSO}[1]{\costDSOfunction(#1)}

\newcommand{\costLA}[2]{\costLAfunction{#1}(#2)}
\newcommand{\costI}[2]{\costIfunction{#1}(#2)}

\newcommand{\linearizedlagrangiankfunction}[1]{\xi^{#1}}
\newcommand{\linearizedlagrangiankIfunction}[2]{\linearizedlagrangiankfunction{#1}_{#2}}

\newcommand{\linearizedlagrangiankI}[3]{\linearizedlagrangiankIfunction{#1}{#2}(#3)}
\newcommand{\riskfunction}{r}

\newcommand{\riskDSOfunction}{\riskfunction_0}

\newcommand{\riskLAfunction}[1]{\riskfunction_{#1}}
\newcommand{\riskIfunction}[1]{\riskfunction_{#1}}
\newcommand{\risk}[1]{\riskfunction(#1)}

\newcommand{\riskDSO}[1]{\riskDSOfunction(#1)}

\newcommand{\riskLA}[2]{\riskLAfunction{#1}(#2)}
\newcommand{\riskI}[2]{\riskIfunction{#1}(#2)}
\newcommand{\xall}{\bm{x}}
\newcommand{\hatxall}{\hat{\xall}}
\newcommand{\dxall}{\bm{t}}
\newcommand{\xallsol}{\xall^\ast}
\newcommand{\xDSO}{\xall_0}

\newcommand{\altxall}{\tilde{\xall}}
\newcommand{\altxDSO}{\altxall_0}
\newcommand{\xallk}[1]{\xall^{#1}}
\newcommand{\xDSOk}[1]{\xDSO^{#1}}
\newcommand{\hatxallk}[1]{\hatxall^{#1}}

\newcommand{\xLAs}{\xall_\Aggregator}

\newcommand{\xLA}[1]{\xall_{#1}}
\newcommand{\xLAk}[2]{\xLA{#1}^{#2}}
\newcommand{\xI}[1]{\xall_{#1}}
\newcommand{\xIk}[2]{\xI{#1}^{#2}}
\newcommand{\xIsol}[1]{\xallsol_{#1}}

\newcommand{\hatxI}[1]{\hatxall_{#1}}
\newcommand{\hatxIk}[2]{\hatxI{#1}^{#2}}

\newcommand{\altxLA}[1]{\altxall_{#1}}
\newcommand{\altxI}[1]{\altxall_{#1}}
\newcommand{\Xall}{\mathcal{X}}
\newcommand{\Xallsol}{\Xall^\ast}
\newcommand{\XDSO}{\mathcal{X}_0}

\newcommand{\XLA}[1]{\mathcal{X}_{#1}}

\newcommand{\Zall}{\bm{z}}
\newcommand{\Zallk}[1]{\Zall^{#1}}
\newcommand{\ZDSO}{\Zall_0}
\newcommand{\ZDSOk}[1]{\ZDSO^{#1}}

\newcommand{\ZLA}[1]{\Zall_{#1}}
\newcommand{\ZLAk}[2]{\ZLA{#1}^{#2}}

\newcommand{\Sall}{\bm{s}}
\newcommand{\Sallk}[1]{\Sall^{#1}}
\newcommand{\SDSO}{\Sall_0}
\newcommand{\SDSOk}[1]{\SDSO^{#1}}

\newcommand{\SLA}[1]{\Sall_{#1}}
\newcommand{\SLAk}[2]{\SLA{#1}^{#2}}

\newcommand{\thetak}[1]{\theta_{#1}}
\newcommand{\varthetak}[1]{\vartheta_{#1}}
\newcommand{\dualvariable}{y}
\newcommand{\altdualvariable}{u}

\newcommand{\LAaltdualvariable}{v}
\newcommand{\LAaltaltdualvariable}{w}
\newcommand{\dualk}[1]{\dualvariable^{#1}}
\newcommand{\altdualk}[1]{\altdualvariable^{#1}}

\newcommand{\LAaltdualk}[1]{\LAaltdualvariable^{#1}}
\newcommand{\LAaltaltdualk}[1]{\LAaltaltdualvariable^{#1}}
\newcommand{\LAaltdualLA}[1]{\altdualvariable_{#1}}
\newcommand{\dualvariablecn}[2]{\bm{\dualvariable}^{\textup{#1}}_{#2}}

\newcommand{\dlmqn}[1]{\dualvariablecn{q}{#1}}

\newcommand{\dlmpn}[1]{\dualvariablecn{p}{#1}}

\newcommand{\constraintfunction}{h}
\newcommand{\constraintmin}{\constraintfunction^\ast}
\newcommand{\constraint}[1]{\constraintfunction(#1)}

\newcommand{\gradI}[1]{\grad_{#1}}
\newcommand{\cur}{l}
\newcommand{\vcur}{\bm{\cur}}
\newcommand{\curn}[1]{\vcur_{#1}}

\newcommand{\curnt}[2]{\cur_{#1,#2}}
\newcommand{\vol}{v}
\newcommand{\vvol}{\bm{\vol}}
\newcommand{\voln}[1]{\vvol_{#1}}

\newcommand{\volnt}[2]{\vol_{#1,#2}}
\newcommand{\bvol}{V}
\newcommand{\ubvol}{\bar{\bvol}}
\newcommand{\ubvoln}[1]{\ubvol_{#1}}
\newcommand{\lbvol}{\underline{\bvol}}
\newcommand{\lbvoln}[1]{\lbvol_{#1}}

\newcommand{\apf}{f}
\newcommand{\vapf}{\bm{\apf}}
\newcommand{\apfn}[1]{\vapf_{#1}}

\newcommand{\apfnt}[2]{\apf_{#1,#2}}
\newcommand{\rpf}{g}
\newcommand{\vrpf}{\bm{\rpf}}
\newcommand{\rpfn}[1]{\vrpf_{#1}}

\newcommand{\rpfnt}[2]{\rpf_{#1,#2}}
\newcommand{\ubpf}{S}
\newcommand{\ubpfn}[1]{\ubpf_{#1}}
\newcommand{\napc}{p}
\newcommand{\vnapc}{\bm{\napc}}
\newcommand{\napcn}[1]{\vnapc_{#1}}

\newcommand{\napcnt}[2]{\napc_{#1,#2}}
\newcommand{\apc}{\napc^{\consumed}}

\newcommand{\apcnt}[2]{\apc_{#1,#2}}
\newcommand{\lbmapc}{E}
\newcommand{\lbmapcn}[1]{\lbmapc_{#1}}
\newcommand{\bapc}{P}
\newcommand{\ubapc}{\overline{\bapc}}
\newcommand{\ubapcnt}[2]{\ubapc_{#1,#2}}
\newcommand{\ubapcn}[1]{\bm{\ubapc}_{#1}}
\newcommand{\lbapc}{\underline{\bapc}}
\newcommand{\lbapcnt}[2]{\lbapc_{#1,#2}}
\newcommand{\lbapcn}[1]{\bm{\lbapc}_{#1}}
\newcommand{\ratiopc}{\tau^{\consumed}}
\newcommand{\ratiopcn}[1]{\ratiopc_{#1}}
\newcommand{\nrpc}{q}
\newcommand{\vnrpc}{\bm{\nrpc}}
\newcommand{\nrpcn}[1]{\vnrpc_{#1}}

\newcommand{\nrpcnt}[2]{\nrpc_{#1,#2}}
\newcommand{\rpc}{\nrpc^{\consumed}}

\newcommand{\rpcnt}[2]{\rpc_{#1,#2}}
\newcommand{\app}{\napc^{\produced}}

\newcommand{\appnt}[2]{\app_{#1,#2}}

\newcommand{\lbratioapp}{\underline{\rho}^{\produced}}
\newcommand{\ubratioapp}{\overline{\rho}^{\produced}}
\newcommand{\lbratioappnt}[2]{\lbratioapp_{#1,#2}}
\newcommand{\ubratioappnt}[2]{\ubratioapp_{#1,#2}}
\newcommand{\rpp}{\nrpc^{\produced}}

\newcommand{\rppnt}[2]{\rpp_{#1,#2}}
\newcommand{\res}{R}
\newcommand{\resn}[1]{\res_{#1}}
\newcommand{\rea}{X}
\newcommand{\rean}[1]{\rea_{#1}}
\newcommand{\con}{G}
\newcommand{\conn}[1]{\con_{#1}}
\newcommand{\sus}{B}
\newcommand{\susn}[1]{\sus_{#1}}

 \newcommand{\UnitI}[1]{U_{#1}}
 
 \renewcommand{\identity}{\bm{I}}

 \newcommand{\identityI}[1]{\identity_{#1}}

\newcommand{\mixedcriterionfunction}{R^{\textup{KKT}}}\newcommand{\mixedcriterion}[2]{\mixedcriterionfunction({#1,#2})}
\newcommand{\lagrangianfunction}{L}
\newcommand{\lagrangian}[2]{\lagrangianfunction(#1,#2)}

\renewcommand{\Xall}{\Real^\primaldimension}%
\renewcommand{\Xallsol}{\mathcal{X}^\ast}%
\renewcommand{\block}{\bm{i}}
\renewcommand{\identity}{I}
\renewcommand{\indicatorfunction}[1]{\bm{1}_{#1}}

\newcommand{\hhm}{\hspace{-2pt}}
 
\newacro{PDA}{primal-dual algorithm}
\newacro{SOCP}{Second-Order Conic Programming}
\newacro{DSO}{Distribution System Operator}
\newacro{LA}{Load Aggregator}
\newacro{OPF}{Optimal Power Flow}
\newacro{AC}{alternative current}
\newacro{DER}{Distributed Energy Resource}
\newacro{LMP}{Locational Marginal Price}
\newacro{DLMP}{Distribution Locational Marginal Price}
\newacro{ACOPF}{Alternative Current Optimal Power Flow}
\newacro{PPDLMP}{privacy-preserving DLMP solver}

\title{\LARGE \bf
A Privacy-preserving Decentralized Algorithm for Distribution Locational Marginal Prices 
}
\date{\today}

\author{Olivier Bilenne$^{1}$, Paulin Jacquot$^{2}$, Nadia Oudjane$^{2}$, Mathias Staudigl$^{1}$, Cheng Wan$^{2}$ and Barbara Franci$^{1}$
\thanks{*This research benefited from the support of the FMJH Program PGMO and from the support of EDF.}

\thanks{$^{1}$ Olivier Bilenne, Barbara Franci and Mathias Staudigl are with the Department of Advanced Computing Sciences, Maastricht University, The Netherlands,
{\tt\small o.bilenne;m.staudigl@maastrichtuniversity.nl}}
\thanks{$^{2}$ Paulin Jacquot, Nadia Oudjane and Cheng Wan are with EDF R\&D, OSIRIS, Palaiseau, France.
{\tt\small paulin.jacquot;nadia.oudjane;cheng.wan@edf.fr } }
}

\begin{document}
\maketitle
\thispagestyle{empty}
\pagestyle{empty}

\begin{abstract}

A major challenge in today’s electricity system is the management of flexibilities offered by new usages, such as smart home appliances or electric vehicles. By incentivizing energy consumption profiles of individuals, demand response seeks to adjust the power demand to the supply, for increased grid stability and better integration of renewable energies. This optimization of flexibility is typically managed by Load Aggregators, independent entities which aggregate and optimize numerous flexibility providers. The consideration of the underlying distribution network constraints,  which couple the different actors, leads to a complex multi-agent problem.  To address it, we propose a new decentralized algorithm that solves a convex relaxation of the classical Alternative Current Optimal Power Flow (ACOPF) problem, and which relies on local information only. Each computational step is performed in a privacy-preserving manner, and system-wide coordination is achieved via node-specific distribution locational marginal prices (DLMPs). We demonstrate the efficiency of our approach on a 15-bus radial distribution network.
\end{abstract}

\renewcommand{\sharp}{\gamma}
\acresetall
\allowdisplaybreaks

\newcommand{\eqd}{\overset{\text{def}}{=}} 
\section{Introduction}
\label{sec:intro}
The modern distribution network is undergoing an unprecedent reformation, thanks to the increased deployment of \acndefp{DER} in the form of distributed generators, distributed storage, microgrids, aggregators managing fleets of electric vehicles or groups of prosumers \cite{NAPS16}.
While the potential benefits of \acp{DER} are globally accepted, reaching those benefits requires smart management methods. Specifically, wrong control strategies could lead to drastic voltage fluctuations and supply-demand imbalances.
With this in mind, a replication of the transmission-level \acdef{LMP} (defined as the marginal cost induced by an additional unit of demand at a particular bus) is much desired.  The price signals differ spatially and temporally, and are used to incentivize \acp{DER} to balance supply-demand, support voltage, and minimize system losses. The necessary extension to \emph{Distribution} \acp{LMP}\textemdash abbreviated as \acsp{DLMP}\textemdash has been developed in \cite{SotVig06,Hey12}. A key question in the \ac{DLMP} approach is their effective computation. According to \cite{Con12,SinGos10}, the procedure of using \acp{DLMP} is as follows: the \acdef{DSO} obtains the flexible demand and supply data, such as active and reactive power generation at the buses, from the \acdefp{LA}.
Having complete information about the distribution network and the predicted spot prices at the relevant distribution buses, the \acp{DLMP} are calculated by solving a network optimization problem.
Specifically, \acp{DLMP} are obtained as dual variables, measuring the sensitivity of the network flow constraints describing the physics of the problem, and are then announced to aggregators.
Considering the received \acp{DLMP} and predicted spot prices, each aggregator make its optimal plans and submits its energy schedules to the spot market. Our distributed coordination mechanism follows this approach line-by-line. We develop a new distributed block-coordinate descent algorithm, designed to effectively compute \acp{DLMP} in a decentralized and privacy-preserving way. Our computational architecture involves direct communication between the \acp{LA} and the feeding bus, who acts as a central computational unit which updates the \acp{DLMP}. However, the data communicated to the center are not containing any information on the local cost function or power profiles managed by the \ac{LA}. In that sense, our notion of privacy should not be confused with the influential concept of \emph{differential privacy} in computer science. Our main theoretical result (\Cref{theorem:convergence}) states the convergence of a general \textit{primal-dual block-coordinate descent algorithm}, that extends recent block-coordinate primal-dual splitting methods for solving linearly constrained composite convex optimization as presented in \cite{malitsky17,luke18,malitsky2019primaldual}. 
\subsection{Related Literature}
The increased importance of effective management of \acp{DER} has been supported by an active literature on  decentralized control strategies. \cite{Li:2014aa} introduced \acp{DLMP} for the distributed management of fleets of electric vehicles. \cite{Hua2014} proposed a quadratic programming approach to solve \acp{DLMP} for decentralized congestion management. The Optimal Power Flow model (OPF) is the standard approach for power flow analysis and optimization of power systems. Using the convex relaxation derived in \cite{farivar13} (see also \cite{Kocuk:2016aa}), Papavasiliou \cite{papavasiliou18} derived \acp{DLMP} based on the KKT conditions, yet he does not provide any algorithm to effectively compute \acp{DLMP}.
\cite{Alsaleh:2018aa} proposes to compute \acp{DLMP} via semi-definite programming but their algorithm is not distributed. \cite{Bai:2018aa} computes \acp{DLMP} by focusing on the day-ahead distribution level electricity market. As in our paper, \acp{DLMP} are decomposed into a number of components (i.e., marginal costs for active power, reactive power, congestion, voltage support, and loss), which provide price signals to motivate \acp{DER} to contribute to congestion management and voltage support.
\cite{zhou2017incentive} proposes a practical distributed algorithm for optimizing DERs, but their approach differs from us in that it relies on a linearization of \ac{OPF}, and uses dual decomposition and gradient descent, which requires the exchange of local primal variables.



\section{Distributed Optimal Power Flow}
\label{sec:OPF}
Consider a power system with $\nodenb$ buses $\Node=\{1,\dots,\nodenb\}$ on a radial distribution network, modeled as a tree graph $\mathcal{G}=(\Nodep,\mathcal{E})$, where $\Nodep=\Node\cup\{0\}$. The root node $0$ is selected as the reference bus. The network is optimized over a time window $\Tim=\{1,\dots,\timlast\}$.

\subsection{Branch flow equations}
We use $\napcn{\node}=(\napcnt{\node}{1},\dots,\napcnt{\node}{\timlast})$ and $\nrpcn{\node}=(\nrpcnt{\node}{1},\dots,\nrpcnt{\node}{\timlast})$ to denote active and reactive power consumption at bus $n$ at each time point $t\in\Tim$. Thus, $\napcnt{\node}{\tim}<0$ means that there is production of energy at bus $n$ at time $t$. At $\node=0$, we assume that power will only be generated and there is no consumption, i.e. $\napcnt{0}{\tim}\leq 0$. 

In deriving the power flow equations, we follow \cite{peng18}. Specifically, after elimination of  phase angles and convex relaxation, the \acs{AC} branch flow equations for a node $\node\in\Node$ and its (unique) ancestor on the graph, denoted by~$\nodem$, are:
\begin{subequations}\label{branchflowconstraints}
\newcommand{\spl}{}%
\begin{align}
 \label{flowconservationactive}
\spl & \apfn{\node} \,{-}\, \!\!\sum\limits_{\altnode:\altnodem=\node}\!\! ( \apfn{\altnode} \,{-}\, \resn{\altnode}\curn{\altnode} ) \,{+}\, \napcn{\node} \,{+}\, \conn{\node}\voln{\node} \,{=}\, 0
\hspace{-20mm} &[\dlmpn{\node}]
  \\
  \label{flowconservationreactive} 
\spl & \rpfn{\node} \,{-}\, \!\!\sum\limits_{\altnode:\altnodem=\node}\!\! ( \rpfn{\altnode} \,{-}\, \resn{\altnode}\curn{\altnode} ) \,{+}\, \nrpcn{\node} \,{-}\, \susn{\node}\voln{\node} \,{=}\, 0 \hspace{-20mm}& [\dlmqn{\node}]
  \\
   \label{Ohm} 
\spl &\voln{\node} \,{-}\, 2 (\resn{\node}\apfn{\node} \,{+}\, \rean{\node}\rpfn{\node}) \,{+}\, (\resn{\node}^2 \,{+}\, \rean{\node}^2) \,\curn{\node} \,{=}\, \voln{\nodem} 
\hspace{-4mm}& 
  \\
 \label{powerflowrelaxed}
\spl &  \apfnt{\node}{\tim}^2 +\rpfnt{\node}{\tim}^2 \leq \volnt{\node}{\tim} \curnt{\node}{\tim}
   &\forall \tim\in\Tim
  \\
  \label{powerflowbound}
\spl &  \apfnt{\node}{\tim}^2 +\rpfnt{\node}{\tim}^2 \leq \ubpfn{\node}^2
  & \forall \tim\in\Tim 
  \\
  \label{powerflowboundm}
\spl & (\apfnt{\node}{\tim}-\resn{\node}\curnt{\node}{\tim} )^2 +(\rpfnt{\node}{\tim}-\rean{\node}\curnt{\node}{\tim} )^2 \leq \ubpfn{\node}^2 \hspace{-20mm}&  \forall\tim\in\Tim
  \\
 \label{voltagebound}
\spl & \lbvoln{\node}\leq \volnt{\node}{\tim} \leq \ubvoln{\node} ,
  & \forall  \tim\in\Tim
 \end{align}
\end{subequations}  
where
\begin{itemize}
\item 
$\voln{\node}=(\volnt{\node}{1},\dots,\volnt{\node}{\timlast})$ and~$\voln{\nodem}$ are the squared voltage magnitudes at buses~$\node$ and~$\nodem$,
\item
$\curn{\node}$ is the squared current magnitude on branch $(\node,\nodem)$,
\item
$\apfn{\node}$ and $\rpfn{\node}$ are the active and the reactive parts of the power flow over line $(\node,\nodem)$,
\item
$\resn{\node}$ and~$\rean{\node}$ are the resistance and the reactance of branch $(\node,\nodem)$,
\item
$\conn{\node}$ and~$\susn{\node}$ are the line conductance and susceptance at~$\node$.
\end{itemize}
Equation \eqref{flowconservationactive} and \eqref{flowconservationreactive} are the active and reactive flow conservation equations, ~\eqref{Ohm} is an expression of Ohm's law for the branch $(\node,\nodem)$, and~\eqref{powerflowrelaxed} is a SOCP relaxation of the definition of the power flow \cite{molzahn2019survey}. Equations~\eqref{powerflowbound} and  \eqref{powerflowboundm} are limitations on the  squared power flow magnitude on $(\node,\nodem)$, and~\eqref{voltagebound} gives lower and upper bounds on the voltage at~$\node$. For the coupling flow conservation laws, dual variables are attached, which are the DLMPs corresponding to active and reactive power. There exist theoretical sufficient conditions under which the relaxation \eqref{branchflowconstraints} is exact  \cite{farivar13,molzahn2019survey}.

For later reference, we point out that the network flow constraints \eqref{flowconservationactive}-\eqref{flowconservationreactive} can be compactly summarized as
$$ A_{0}\xDSO+\textstyle\sum_{a}A_{a}\xLA{\aggregator}=b$$ for suitably defined matrices $A_{0},A_{a}$ and  vector $b$.

\subsection{Load aggregators\acused{LA}}
The set of buses~$\Node$ is partitioned into a collection~$(\Nodea{\aggregator})_{\aggregator\in\Aggregator}$ of subsets, such that each node subset~$\Nodea{\aggregator}$ is managed by a \acl{LA}~$\aggregator\in\Aggregator$. Each \ac{LA} controls the flexible net power consumption ($p_{n,t}$) and generation at each node $\node\in\Nodea{\aggregator}$, given at time~$\tim$ by
\begin{subequations}\label{eq:aggregatorconstraints}
\begin{equation}\label{netactivepowerconsumed}
 \napcnt{\node}{\tim} = \apcnt{\node}{\tim} -\appnt{\node}{\tim}
 ,\qquad 
 \nrpcnt{\node}{\tim} = \rpcnt{\node}{\tim} -\rppnt{\node}{\tim}
 ,
\end{equation}
\noeqref{netactivepowerconsumed}%
for all $\node\in\Node$ and $\tim\in\Tim.$ $\apcnt{\node}{\tim}\geq 0$ is the consumption part and $\appnt{\node}{\tim}\geq 0$ is the production part of the power profile. Power consumption and production at the nodes are made flexible by the presence of deferrable loads (electric vehicles, water heaters) and \acdefp{DER}. The consumption at each node $\node\in\Node$ must satisfy a global energy demand $\lbmapcn{\node}$ over the full time window,
\begin{equation} \label{totalactivepowerconsumed}
 \sum_{\tim\in\Tim}\apcnt{\node}{\tim}\geq \lbmapcn{\node}
 , \qquad \forall \node\in\Node.
\end{equation}
\noeqref{totalactivepowerconsumed}%
Consumption and production are also constrained by power bounds and active to reactive power ratio:
\begin{align}\label{boundactivepowerconsumed}
& \lbapcnt{\node}{\tim} \leq \apcnt{\node}{\tim}\leq \ubapcnt{\node}{\tim}
 , &\forall \node\in\Node,\ \forall \tim\in\Tim,
\\
\label{reactivepowerconsumed} 
&\rpcnt{\node}{\tim} = \ratiopcn{\node}\apcnt{\node}{\tim}
 , & \forall \node\in\Node,\ \forall \tim\in\Tim,
 \\
\label{activepowerproduced}
&0 \leq \appnt{\node}{\tim} \leq \ubapc
 , & \forall \node\in\Node,\ \forall \tim\in\Tim,
\\
\label{reactivepowerproduced}
&\lbratioappnt{\node}{\tim}\appnt{\node}{\tim} \leq \rppnt{\node}{\tim} \leq \ubratioappnt{\node}{\tim}\appnt{\node}{\tim}
 , & \forall \node\in\Node,\ \forall \tim\in\Tim.
\end{align}
\noeqref{boundactivepowerconsumed}%
\noeqref{reactivepowerconsumed}%
\noeqref{activepowerproduced}%
\noeqref{reactivepowerproduced}%
\end{subequations}
Constraints \eqref{netactivepowerconsumed}-\eqref{reactivepowerproduced} define the feasible set $\XLA{\aggregator}$ of \ac{LA} decisions, containing vectors $\xLA{\aggregator}=(\napcn{\node},\nrpcn{\node})_{\node\in\Nodea{\aggregator}}$. 
\begin{remark}
We focus on the simple model of power profile constraints given by \eqref{eq:aggregatorconstraints}, which is well  adapted for some flexible electric appliances such as electric vehicles, and has been largely considered in the literature. Yet, as shown in \Cref{sec:genericAlgo}, our method and results apply to a much larger framework.
\end{remark}

Both, consumption and production, must be scheduled by the \ac{LA}, taking into account the current spot market prices, and other specific local factors characterizing the private objectives of the \ac{LA}. Formally, there is a convex cost function $\costLA{\aggregator}{\xLA{\aggregator}}$ which the \ac{LA} would like to unilaterally minimize, subject to private feasibility $\xLA{\aggregator}\in\XLA{\aggregator}$.

\subsection{The distribution system operator}
In order to guarantee stability of the distribution network, the \ac{DSO} takes the individual aggregators' decisions into account and adjusts the power flows so that the flow conservation constraints \eqref{flowconservationactive}-\eqref{flowconservationreactive}, together with the SOCP constraints \eqref{Ohm}-\eqref{voltagebound}, are satisfied. Let $\xDSO=(\napcn{0},\nrpcn{0},\vapf,\vrpf,\vvol,\vcur)$ denote the vector of the variables controlled by the \ac{DSO}, and define the \ac{DSO}'s feasible set $\XDSO = \{ \xDSO \vert \eqref{Ohm}-\eqref{voltagebound}\text{ hold for }\node\in\Node \}.$ Then, the set of \ac{DSO} decision variables inducing a physically meaningful network flow for a given tuple of \ac{LA} decisions $\xLAs$ is described as 
\begin{equation*}\label{eq:DSOfeas}
\mathcal{F}(\xLAs)=\{\xDSO\in\XDSO\vert\eqref{flowconservationactive}-\eqref{flowconservationreactive}\text{ hold for }\xLAs\}.
\end{equation*}
Denoting the \ac{DSO} cost function $\phi_{0}(\xDSO)$, we arrive at the \ac{DSO}'s decision problem 
\begin{equation}\label{eq:Psi}
\Psi(\xLAs)=\min\{\phi_{0}(\xDSO)\vert \xDSO\in\mathcal{F}(\xLAs)\},
\end{equation}
This represents the smallest costs to the \ac{DSO}, given the profile of flexible net consumption and generation at each affiliated node $n\in\Node_{\aggregator}$.
\section{Privacy-preserving DLMP Computation}
\label{sec:Algo}
Ww are facing a multi-agent optimization problem, in which \acp{LA} and a singe DSO aim for solving the AC-OPF problem by unilaterally solving their individual cost minimization problem. All these decision problems are coupled by the network flow constraints \eqref{flowconservationactive}-\eqref{flowconservationreactive}. 
%
%
Algorithm~\ref{algorithm:primaldual} proposes  a \acdef{PPDLMP}, in which the \ac{DSO} influences the decentralized decisions of the \acp{LA} by sending out information about prevailing \acp{DLMP}, and iteratively updates \acp{DLMP} based on the power profiles in the local markets. \\ 
%
\ac{PPDLMP} asks the \ac{DSO} to adjust \acsp{DLMP} based on the prevailing plans reported by the \acp{LA}. Once the price update is completed, a single \ac{LA} is appointed at random to adapt the power profile within the subnetwork this \ac{LA} manages. 
The local update of the \ac{LA} results in bid vector $w^{k}$, which will be fed into the \ac{DSO} final computational step to perform dispatch. Hence, \ac{PPDLMP} is based on \emph{block-coordinate primal updates}, involving pairs of the type $(\xDSO,\xLA{\aggregator})$ picked randomly with probability $1/\cardinality{\Aggregator}$ for every $\aggregator\in\Aggregator$. 


\smallskip
\begin{algorithmenv}[t]
\caption{\acl{PPDLMP} (PPDLMP)}
\label{algorithm:primaldual}
\removelatexerror
\RestyleAlgo{tworuled}%
\centering
\vspace{-2.5pt}
\begin{algorithm}[H]
\setcounter{AlgoLine}{0}
\DontPrintSemicolon
\SetAlgoNoLine%
\SetKwInOut{Parameters}{\texttt{\textbf{Parameters}}}
\SetKw{Output}{\texttt{\textbf{Output:}}}
\SetKwFor{Initatdo}{\texttt{Initialization at}}{\texttt{:}}{}
\SetKw{phantomInit}{\phantom{\textbf{Initialization:}}}
\SetKwFor{For}{\texttt{for}}{\texttt{do}}{}
\SetKwFor{Atdo}{\texttt{at}}{\texttt{do}}{}
\Parameters{$\nbaggregators=\cardinality{\Aggregator}$, $\dualstepsize>0$, \reviseagain{$\StepsizeDSOk{\iter}$, $\StepsizeLAk{\Aggregator}{\iter}$}
}
\nonl\Initatdo{\texttt{\textup{\textbf{each aggregator $\aggregator\in\Aggregator$}}}}{
\nonl $\xLAk{\Aggregator}{0}\old{=\SLAk{\Aggregator}{0}}\in\XLA{\aggregator}$\;
\nonl  \texttt{\textbf{send bid}}~$\LAaltdualLA{\aggregator}= \ALA{\aggregator}\xLAk{\aggregator}{0}  - \bvectorLA{\aggregator}$ \texttt{\textbf{to the  \ac{DSO}}}
}
\nonl\Initatdo{\texttt{\textbf{the \ac{DSO}}}}{ 
\nonl $\xDSOk{0}\old{=\SDSOk{0}}\in \XDSO$\;
\nonl $\LAaltdualk{0}= \dualstepsize \sum_{\aggregator\in\Aggregator}\LAaltdualLA{\aggregator} 
$, \quad $\dualk{0}= \LAaltdualk{0}+\dualstepsize  (\ADSO \xDSOk{0} -\bvectorDSO)$\;
}
\nonl\Output{\reviseagain{$\xallk{\iter}$, $\Sallk{\iter}=\frac{1}{\iter}\sum_{\altiter=1}^{\iter}\xallk{\altiter}$
}}\;
\smallskip
\SetAlgoLined
 \nonl \For{$\iter=0,1,2,\dots$}{
 \nonl \Atdo{\texttt{\textbf{the \ac{DSO}}}}{
\old{
\routinenl
$\ZDSOk{\iter} = (\frac{\iter}{\iter+1})\, \SDSOk{\iter} + (\frac{1}{\iter+1}) \,\xDSOk{\iter} $ \;
}
\routinenl
$\xDSOk{\iter+1} = \arg\min\limits_{\altxDSO\in\XDSO}\big\{ \inner{\grad\costDSO{\reviseagain{\xDSOk{\iter}}} + \ADSO^\transpose \dualk{\iter}}{\altxDSO}$\;\nonl \qquad\qquad\qquad\qquad\qquad\qquad$ + \frac{1}{2} \normx{\altxDSO-\xDSOk{\iter}}{\StepsizeDSOk{\iter}}^2 \big\} $
 \; \label{algorithm:PGDSO}
\old{
\routinenl$\SDSOk{\iter+1}  =   \ZDSOk{\iter}    +   \frac{1}{\iter+1} (\xDSOk{\iter+1} -\xDSOk{\iter}) $ \;%
}

}
\nonl\Atdo{\texttt{\textbf{\ac{LA}~$\aggregator$ drawn uniformly at random}}}{
\old{
\routinenl$\ZLAk{\aggregator}{\iter} = (\frac{\iter}{\iter+1}) \SLAk{\aggregator}{\iter} + (\frac{1}{\iter+1}) \xLAk{\aggregator}{\iter} $ \; }
\routinenl\texttt{\textbf{receive \acs{DLMP}}}~$\dualk{\iter}$ \texttt{\textbf{from \ac{DSO}}} \;
\label{algorithm:PGLA}%
\routinenl$ \xLAk{\aggregator}{\iter+1} = \arg\min\limits_{\altxLA{\aggregator}\in\XLA{\aggregator}}\big\{ \inner{\grad\costLA{\aggregator}{\reviseagain{\xLAk{\aggregator}{\iter}}} +\ALA{\aggregator}^\transpose \dualk{\iter}}{\altxLA{\aggregator}}
$ \; 
\nonl \qquad\qquad\qquad\qquad\qquad\qquad$ + \frac{\nbaggregators}{2} \normx{\altxLA{\aggregator}-\xLAk{\aggregator}{\iter}}{\StepsizeLAk{\aggregator}{\iter}}^2 \big\}
$ \; 
\old{
\routinenl$\SLAk{\aggregator}{\iter+1}  =   \ZLAk{\aggregator}{\iter}    +   \frac{\nbaggregators}{\iter+1} (\xLAk{\aggregator}{\iter+1} -\xLAk{\aggregator}{\iter}) $ \;%
}
\routinenl$ \LAaltaltdualk{\iter} = \ALA{\aggregator} (\xLAk{\aggregator}{\iter+1}-\xLAk{\aggregator}{\iter}) $ \;
}
\nonl\Atdo{\texttt{\textbf{each other aggregator~$\altaggregator\neq\aggregator$}}}{
\routinenl $\xLAk{\altaggregator}{\iter+1} = \xLAk{\altaggregator}{\iter} $ \;
\old{
\routinenl$\SLAk{\altaggregator}{\iter+1}  =   
(\frac{\iter}{\iter+1})\, \SLAk{\altaggregator}{\iter} + (\frac{1}{\iter+1})\, \xLAk{\altaggregator}{\iter}
$ \;%
}
}
 \nonl \Atdo{\texttt{\textbf{the \ac{DSO}}}}{
\routinenl \texttt{\textbf{receive bid}}~$\LAaltaltdualk{\iter} $ \texttt{\textbf{from \ac{LA}~$\aggregator$}} \;
\routinenl$\dualk{\iter+1} \hhm = \hhm \dualk{\iter}  \hhm+ \hhm \dualstepsize  [\ADSO (2\xDSOk{\iter+1}\hhm-\hhm\xDSOk{\iter}) - \bvectorDSO] \hhm+\hhm \LAaltdualk{\iter} 
\hhm +\hhm \dualstepsize (\nbaggregators \hhm+\hhm 1) \LAaltaltdualk{\iter}  $  \label{algorithm:dualk}\; 
\routinenl$\LAaltdualk{\iter+1} =  \LAaltdualk{\iter} + \dualstepsize  \LAaltaltdualk{\iter}  $  \;
}
}
 \end{algorithm}
\end{algorithmenv}
It is important to point out that, while executing \ac{PPDLMP}, the bus-specific data (like cost function, power profiles,etc.) \textit{remain private information}. This applies equally to the \ac{DSO} and the \ac{LA}. Coordination of the system-wide behavior is achieved via exchanging information about \emph{dual variables} only, describing the DLMPs and the expressed bids of the \ac{LA}s. In that sense, \ac{PPDLMP} describes a \emph{semi-distributed} multi-agent optimization scheme. 

\section{Primal-dual Block Coordinate Descent}
\label{sec:genericAlgo}
We study the convergence properties of \ac{PPDLMP} via the analysis of a more general and new block-coordinate descent method designed to solve composite convex optimization problems of the form 
\begin{equation}\label{genericP}
\begin{array}{ll} 
  \minimize\limits_{\xall\in\Xall} & 
  \reviseagain{ \left\{ \cost{\xall} + \risk{\xall} = \extcost{\xall} \right\} }
  \\ 
  \subjectto &   \xall\in\arg\min_{\altxall}\constraint{\altxall}
  .
 \end{array}
\end{equation}
where $\constraint{\xall} = \tfrac 1 2 \norm{\Aall\xall-\bvector}^{2},$ in which $\Aall\in\Real^{\dualdimension\times\primaldimension}$ and  $\bvector\in\Real^{\dualdimension}$. We assume that the decision variable is partitioned into $d$ blocks $\xall=(\xI{1},\dots,\xI{\dimension})^{\top}$ with $\xI{\coordinate}\in\Real^{\dimensioni{\coordinate}}$ and $ \sum_{\coordinate=1}^\dimension\dimensioni{\coordinate}=\primaldimension$. The separable cost function $\cost{\xall} = \sum_{\coordinate} \costI{\coordinate}{\xI{\coordinate}}$ is convex and smooth in each block. The non-smooth component $ \risk{\xall} =\sum_{\coordinate} \riskI{\coordinate}{\xI{\coordinate}}
$  is additively separable with respect to the~$\dimension$ block-coordinate directions, and write $\Aall=(\AI{1} \hdots \AI{\dimension})$ with $\AI{\coordinate}\in\Real^{\dualdimension\times\dimensioni{\coordinate}}$ for $\coordinate=1,\dots,\dimension$. We assume that $\riskfunction:\Real^\primaldimension\to (-\infty,+\infty]$ is a proper closed lower semi-continuous and prox-friendly function. In order to recover the OPF problem, we identify each function $\phi_{i}$ with a cost function of the \ac{DSO} or \ac{LA}, and $r_{i}$ is an indicator function of the feasible set $\XLA{\aggregator}$ and $\XDSO$, respectively. We also assume that there exists a positive semidefinite matrix $\Smoothness\in\Real^{m\times m}$ such that, for every $\xall,\altxall\in\reviseagain{\domain{\riskfunction}}
$, it holds that
\begin{equation}\label{smoothness}
 \cost{\altxall} \leq \cost{\xall} +   \inner{\grad\cost{\xall}}{\altxall-\xall}+ \tfrac{1}{2}\normx{\altxall-\xall}{\Smoothness}^2,
\end{equation} 
where $\normx{\cdot}{\Smoothness}\eqd\sqrt{\inner{\Smoothness\cdot}{\cdot}}$. If $\Lambda=\diag(\lambda_{1}\identityI{m_{1}},\ldots,\lambda_{d}\identityI{m_{d}})$, \eqref{smoothness} reduces to the well-known descent lemma for smooth functions with a Lipschitz continuous gradient \cite{Nes18}.  

Our approach is a block-coordinate implementation of the method developed in~\cite{malitsky2019primaldual} for linearly constrained optimization, lying midway between the celebrated Chambolle-Pock primal-dual splitting algorithm~\cite{chambolle11} and Tseng's accelerated proximal gradient~\cite{tseng08apgm}.
The present setting differs from~\cite{luke18}'s coordinate-descent interpretation of~\cite{malitsky2019primaldual} in that composite objective functions are considered, and block sampling is used for the coordinates. Precisely, we consider a set $\Block$ of blocks such that:
$$\forall \block \in \Block, \block \subset \{0\}\times \{1, \dots , d \} \ \text{ and }\  \textstyle\bigcup_{\block \in \Block} \block = \{0,1, \dots ,d \} \ .$$
Define the $m\times m$ weighting matrices  
$\Dimension\eqd \diag(1/\probai{1}\identityI{m_{1}},\dots,1/\probai{\dimension}\identityI{m_{p}})$ and $\DimensionI{\block}\eqd \diag\lbrack (1/\probai{\coordinate}\identityI{m_{i}})_{\coordinate\in\block} \rbrack$ and  for all $\block\in\Block$. Similarly, let $\Stepsizek{\iter} \eqd\diag(\Stepsizeki{\iter}{1},\dots,\Stepsizeki{\iter}{\dimension}) 
\succ 0$ be a block diagonal matrix and, for each $\block\in\Block$, define $\StepsizekI{\iter}{\block} \eqd \diag\lbrack (\Stepsizeki{\iter}{\coordinate})_{\coordinate\in\block} \rbrack$. If coordinate $i$ is selected for updating, a proximal-based update step, based on the linearization $
\linearizedlagrangiankI{\iter}{\coordinate}{\altxI{\coordinate}} \eqd \inner{\grad\costI{\coordinate}{\xIk{\coordinate}{\iter}}+\AI{\coordinate}^\transpose\dualk{\iter}}{\altxI{\coordinate}}$, is performed in parallel. This delivers the next iterate 
\[
x_{i}^{k+1}=\arg\min_{u_{i}}\{\xi^{k}_{i}(u_{i})+r_{i}(u_{i})+\frac{1}{2}\norm{u_{i}-x_{i}^{k}}^{2}_{P_{i}T_{i}}\}.
\]
\newcommand{\hideiftight}[1]{}%
\begin{algorithmenv}[t]
\caption{Primal-dual Block Coordinate Descent
\label{algorithm:genericprimaldual}}
\removelatexerror
\RestyleAlgo{tworuled}%
\centering
\vspace{-2.5pt}
\begin{algorithm}[H]
\setcounter{AlgoLine}{0}
\DontPrintSemicolon
\SetAlgoNoLine%
\SetKwInOut{Parameters}{{\textbf{Parameters}}}
\SetKwInOut{Init}{{\textbf{Initialization}}}
\SetKwInOut{Output}{{\textbf{Output}}}
\SetKwFor{For}{for}{do}{} 
\Parameters{$\Dimension$, 
$\dualstepsize>0$, $ \Stepsize $, 
$(\thetak{\iter})_{\iter\geq 0}$} 
\Init{$\xallk{0}\old{=\Sallk{0}}\in\Xall$, $\altdualk{0} = \dualstepsize (\Aall \xallk{0}  -  \bvector )$, $ \dualk{0} = (\dualstepsize/\thetak{0}) (\Aall\xallk{0}-\bvector)   $} 
\Output{\reviseagain{$\xallk{\iter}$, $\Sallk{\iter}=\frac{1}{\iter}\sum_{\altiter=1}^{\iter}\xallk{\altiter}$}}
\smallskip 
\SetAlgoLined
 \nonl \For{$\iter=0,1,2,\dots$}{
\old{
\routinenl
$\Zallk{\iter} = (1-\thetak{\iter}) \Sallk{\iter} + \thetak{\iter} \xallk{\iter} $ \; \label{algorithm:genericPDZallk}
}
\hideiftight{
$ \varthetak{\iter+1} = \thetak{\iter}  (1-\thetak{\iter+1}) /\thetak{\iter+1} $ \; \label{algorithm:genericvarthetak}
}%
\routinenl
\text{{\textbf{draw block~$\block\in\Block$ at random according to~$\Probas$} }} \;
\routinenl
\text{}$ \xIk{\block}{\iter+1} = \arg\min_{\altxI{\block}}\big\{ \inner{\grad\costI{\block}{\reviseagain{\xIk{\block}{\iter}}} + \AI{\block}^\transpose\dualk{\iter}}{\altxI{\block}}$\;\nonl
\qquad\qquad\qquad\qquad$
+ \riskI{\block}{\altxI{\block}} + \tfrac{1}{2} \normx{\altxI{\block}-\xIk{\block}{\iter}}{\DimensionI{\block}\StepsizeI{\block}}^2 \big\}
$ \; 
\routinenl
$\xIk{\minusblock}{\iter+1} = \xIk{\minusblock}{\iter} $ \; 
\old{
\routinenl
$\Sallk{\iter+1}  =   \Zallk{\iter}    +   \thetak{\iter} \Dimension (\xallk{\iter+1} -\xallk{\iter}) $ \;\label{algorithm:genericPDSallk}
}
\routinenl
$ \altdualk{\iter+1} = \altdualk{\iter}+ \dualstepsize \Aall ( \xallk{\iter+1}  -  \xallk{\iter} )
$ \; \label{algorithm:genericaltdualk}
\routinenl
$  \dualk{\iter+1} = \hideiftight{\varthetak{\iter+1}} \dualk{\iter} +  \hideiftight{\varthetak{\iter+1}} \dualstepsize \Aall   \Dimension (\xallk{\iter+1} -\xallk{\iter})  + \altdualk{\iter+1}
$ \; \label{algorithm:genericdualk}
}
\end{algorithm}
\end{algorithmenv}%
In Appendix \ref{sec:Algoequivalent} we show that \ac{PPDLMP} (Algorithm \ref{algorithm:primaldual}) is a special case of the more general primal-dual  Algorithm~\ref{algorithm:genericprimaldual}. 

In Algorithm~\ref{algorithm:genericprimaldual}, a sensible choice for~$\Stepsizek{\iter}$ is to set 
 \begin{equation}\label{earlydiagonalscaling}
 \StepsizeI{\coordinate} =
\identityI{\primaldimensioni{\coordinate}}/\stepsizei{\coordinate}
+\probai{\coordinate} \SmoothnessI{\coordinate}
  +\dualstepsize\AI{\coordinate}^\transpose\AI{\coordinate}\quad
(\coordinate=1,\dots,\dimension)
\end{equation}
with the constraint 
$T\succ\dualstepsize \constraintcovariance$, where $(\constraintcovariance)_{{\coordinate}{\altcoordinate}}={\probaij{\coordinate}{\altcoordinate}}\AI{\coordinate}^\transpose\AI{\altcoordinate}/({\probai{\coordinate}\probai{\altcoordinate}})$ and $\probaij{\coordinate}{\altcoordinate}=\proba{\coordinate,\altcoordinate\in\block}$.
A detailed analysis of the sequence generated by Algorithm~\ref{algorithm:genericprimaldual} yields our main result. 
\begin{theorem} \label{theorem:convergence}
Let~$\Xallsol$ denote the solution set of Problem~\eqref{genericP}, and let~$(\xallk{\iter})_{k}$ and~$(\Sallk{\iter})_{k}$  be issued by Algorithm~\ref{algorithm:genericprimaldual} \reviseagain{with $\Sallk{\iter}=\frac{1}{\iter}\sum_{\altiter=1}^{\iter}\xallk{\altiter}$ and} with $\reviseagain{\Stepsizek{\iter}},\dualstepsize$ satisfying~\eqref{earlydiagonalscaling}. Then, 
\begin{enumerate}[(i)]
 \item If there exists a Lagrange multiplier for Problem~\eqref{genericP}, then~$(\xallk{\iter})_{k}$ and~ $(\Sallk{\iter})_{k}$ converge a.s. to a solution of~\eqref{genericP} and $\constraint{\xallk{\iter}}-\constraintmin = \zero{1/\iter}$, $\constraint{\Sallk{\iter}}-\constraintmin = \magnitude{1/\iter^2}$ a.s.
\item If~$\Xallsol$ is a bounded set and~$\costfunction+\riskfunction$ is bounded from below, then a.s. all limit points of~$(\Sallk{\iter})_{k}$ belong to~$\Xallsol$ and $\constraint{\Sallk{\iter}}-\constraintmin = \zero{1/\iter}$.
\end{enumerate}
\end{theorem}
The proof is provided in Appendix \ref{section:convergenceanalysis}.

\section{Numerical Results}
\label{sec:Numerics}
%

\newcommand{\hh}{\hspace{-2pt}}
\newcommand{\T}{\mathcal{T}}
\newcommand{\N}{\mathcal{N}}
\newcommand{\mc}[1]{\mathcal{#1}}
\newcommand{\nt}{_{n,t}}

\newcommand{\xx}{\bm{x}}

\newcommand{\uV}{\underline{V}}
\newcommand{\oV}{\overline{V}}
\newcommand{\oP}{\overline{P}}
\newcommand{\uP}{\underline{P}}
\newcommand{\oQ}{\overline{Q}}
\newcommand{\uQ}{\underline{Q}}
\newcommand{\osigma}{\overline{\sigma}}
\newcommand{\usigma}{\underline{\sigma}}
\newcommand{\odelta}{\overline{\delta}}
\newcommand{\udelta}{\underline{\delta}}
\newcommand{\otheta}{\overline{\theta}}
\newcommand{\utheta}{\underline{\theta}}
\newcommand{\oupsilon}{\overline{\upsilon}}
\newcommand{\uupsilon}{\underline{\upsilon}}
 
\newcommand{\epq}{ e^{pq}}
\newcommand{\oepq}{ \overline{e}^{pq}}
\newcommand{\uepq}{ \underline{e}^{pq}}

\newcommand{\uns}{ \textcolor{blue}{u} } 
\newcommand{\cun}{c^{\text{uns}}} 
\newcommand{\pRE}{p^\mathrm{p}} 
\newcommand{\qRE}{q^\mathrm{p}} 
\newcommand{\pC}{p^\mathrm{c}}
\newcommand{\qC}{q^\mathrm{c}}
\newcommand{\tauC}{\tau^\mathrm{c}}
\newcommand{\rlbRE}{\underline{\rho^\mathrm{p}}} 
\newcommand{\rubRE}{\overline{\rho^\mathrm{p}}} 

\newcommand{\upRE}{\overline{P}^\mathrm{p}} 
\newcommand{\opi}{\overline{\pi}}
\newcommand{\upi}{\underline{\pi}}
\newcommand{\llam}{\bm{\lambda}}
\newcommand{\pbid}{P^{\text{bid}}}

\newcommand{\pp}{\bm{p}}
\newcommand{\qq}{\bm{q}}
\newcommand{\kLoss}{\kappa^{\text{los}}}

\newcommand{\hX}{\hat{\X}}
\newcommand{\A}{\mathcal{A}}
\newcommand{\dlmp}{\lambda^{\text{p}}}
\newcommand{\dlmq}{\lambda^{\text{q}}}
\newcommand{\ddlmp}{\bm{\dlmp}}
\newcommand{\ddlmq}{\bm{\dlmq}}

\newif\ifreferonline
\referonlinefalse

We apply Algorithm \ref{algorithm:primaldual} to a realistic  15-bus network example based on the instance proposed in \cite{papavasiliou18}, over a time horizon $\mathcal{T} = \{0, 1\}$. 
The network parameters are specified in  \Cref{tab:network14param}.
Lines physical parameters $(R_n,X_n,S_n,B_n,V_n)$ are those used in~\cite{papavasiliou18}.
While \cite{papavasiliou18} considers \textit{fixed} loads, here we consider variable, flexible active and reactive loads as specified in \eqref{eq:aggregatorconstraints} and with parameters  $(\lbapcn{\node},\ubapcn{\node},E_n,\tauC_n)_n$  generated  based on the values of  \cite{papavasiliou18}; see also \cite{jacquot20}.

%
Bus $11$ is the only bus to offer renewable production, with $ \bm{\upRE}_{11} \eqd [0.438,0.201]$ and $\rlbRE = \rubRE = 0$ (the renewable production is purely active power). 
The bounds $(\uV_n,\oV_n)$ are set to 0.81 and $1.21$ for each  $n\in\N$, while $V_0=1.0$.

We consider a zero cost function for each \ac{LA} ($\phi_a=0$ for each $a\in\A$), meaning that \acp{LA} are indifferent to  consumption profiles for as long as their feasibility constraints are satisfied.
This is a reasonable assumption in practice for certain types of consumption flexibilities (electric vehicles, batteries).
We consider the \ac{DSO} objective
$$ \phi(\xx)= \phi_0(\xx_0)=  \sum_{t\in\T}c_t( \pRE_{0t}) + k^\text{loss} \sum_{n,t} R_n \ell_{nt},$$
with loss penalization factor $k^\text{loss}=0.001$ and with:
$$c_0:p \mapsto 2p + p^2 , \ c_1 : p \mapsto p , $$
giving an \textit{expensive} time period and a \textit{cheap} one, which can be interpreted as peak and offpeak pricing periods.
\newcommand{\hhh}{\hspace{-3pt}}
\begin{table}[!ht]
 \centering
\begin{scriptsize}
\setlength{\tabcolsep}{1.5pt}
\begin{tabular}{ccccccccc}
\hline
 $n$ &      $S_{n}$ &      $R_{n}{\cdot}10^3$ &      $X_{n}{\cdot}10^3$ &   $B_{n}{\cdot}10^3$ &       $\uP_{n}$ &               $\oP_{n}$ &     $ E_{n} $&    $\tau^{c}_{n}$ \\
\hline
  1 &  2.000 &     1.0 &   120.0 &     1.1 &   [0.593, 0.256] &    [1.566, 1.539] &  2.213 &  0.234 \\
  2 &  0.256 &    88.3 &   126.2 &     2.8 &       [0.000, 0.000] &        [0.000, 0.000] &  0.000 &  0.000 \\
  3 &  0.256 &   138.4 &   197.8 &     2.4 &   [0.003, 0.011] &     [0.020, 0.035] &  0.047 &  0.418 \\
  4 &  0.256 &    19.1 &    27.3 &     0.4 &   [0.015, 0.013] &    [0.027, 0.019] &  0.033 &  0.249 \\
  5 &  0.256 &    17.5 &    25.1 &     0.8 &   [0.021, 0.024] &    [0.043, 0.053] &  0.072 &  0.251 \\
  6 &  0.256 &    48.2 &    68.9 &     0.6 &   [0.017, 0.001] &    [0.032, 0.037] &  0.039 &  0.251 \\
  8 &  0.256 &    40.7 &    58.2 &     1.2 &   [0.021, 0.009] &     [0.040, 0.039] &  0.049 &  0.251 \\
  7 &  0.256 &    52.3 &    74.7 &     0.6 &  [-0.233,-0.210] &  [-0.173,-0.115] & -0.352 &  0.000 \\
  9 &  0.256 &    10.0 &    14.3 &     0.4 &   [0.008, 0.002] &    [0.032, 0.028] &  0.015 &  0.620 \\
 10 &  0.256 &    24.1 &    34.5 &     0.4 &   [0.004, 0.001] &     [0.024, 0.040] &  0.013 &  0.300 \\
 11 &  0.256 &    10.3 &    14.8 &     0.1 &     [0.010, 0.010] &    [0.015, 0.024] &  0.028 &  0.250 \\
 12 &  0.600 &     1.0 &   120.0 &     0.1 &   [0.243, 0.057] &    [0.642, 0.625] &  0.895 &  0.208 \\
 13 &  0.204 &   155.9 &   111.9 &     0.2 &     [0.001, 0.000] &    [0.003, 0.003] &  0.003 &  0.571 \\
 14 &  0.204 &    95.3 &    68.4 &     0.1 &   [0.015, 0.012] &    [0.032, 0.042] &  0.042 &  0.371 \\
\hline
\end{tabular}
\end{scriptsize}
\caption{Parameters for the 15 buses network based on  \cite{papavasiliou18}}
\label{tab:network14param}
\end{table}
%
%
%

The solution obtained by Algorithm~\ref{algorithm:primaldual}
 after 2000 iterations is illustrated in  \Cref{fig:resFlows}, which displays the active flows directions as well as the \ac{DLMP} values.
 
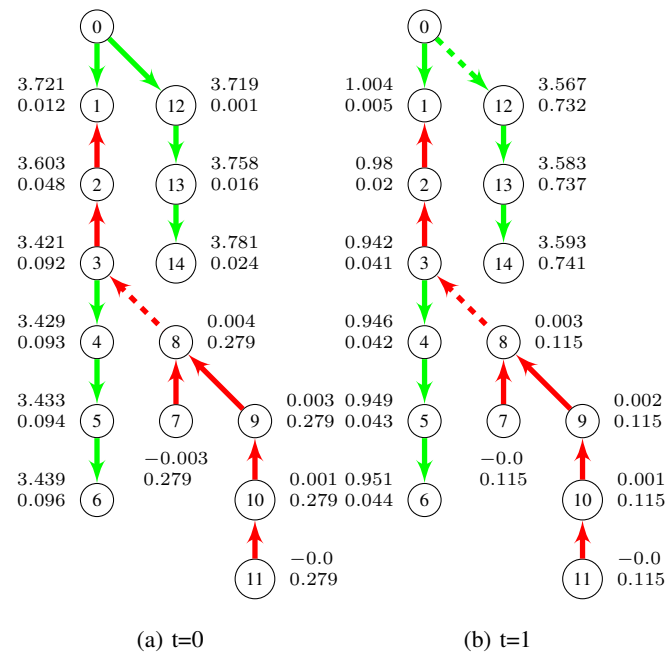
\begin{figure}[h]
\begin{subfigure}{0.49\columnwidth} \centering
  \begin{scriptsize}
\begin{tikzpicture}[scale=0.3]
\node [draw,circle] (node0) at (0,0) {0} ;
\node [draw,circle] (node1) at (0.0, -3.5) {1} ;
 \draw [-latex',green,line width=2pt] (node0) -- (node1); 
  
\node [left= 0.1 cm of node1] (res0) {\pbox{2.5cm}{
$3.721$ \\ 
$0.012$ \\ 
}} ;
\node [draw,circle] (node2) at (0.0, -7.0) {2} ;
 \draw [latex'-, red,line width=2pt] (node1) -- (node2); 
  
\node [left= 0.1 cm of node2] (res0) {\pbox{2.5cm}{
$3.603$ \\ 
$0.048$ \\ 
}} ;
\node [draw,circle] (node3) at (0.0, -10.5) {3} ;
 \draw [latex'-, red,line width=2pt] (node2) -- (node3); 
  
\node [left= 0.1 cm of node3] (res0) {\pbox{2.5cm}{
$3.421$ \\ 
$0.092$ \\ 
}} ;
\node [draw,circle] (node4) at (0.0, -14.0) {4} ;
 \draw [-latex',green,line width=2pt] (node3) -- (node4); 
  
\node [left= 0.1 cm of node4] (res0) {\pbox{2.5cm}{
$3.429$ \\ 
$0.093$ \\ 
}} ;
\node [draw,circle] (node5) at (0.0, -17.5) {5} ;
 \draw [-latex',green,line width=2pt] (node4) -- (node5); 
  
\node [left= 0.1 cm of node5] (res0) {\pbox{2.5cm}{
$3.433$ \\ 
$0.094$ \\ 
}} ;
\node [draw,circle] (node6) at (0.0, -21.0) {6} ;
 \draw [-latex',green,line width=2pt] (node5) -- (node6); 
  
\node [left= 0.1 cm of node6] (res0) {\pbox{2.5cm}{
$3.439$ \\ 
$0.096$ \\ 
}} ;
\node [draw,circle] (node8) at (3.5, -14.0) {8} ;
 \draw [latex'-, red, style = dashed ,line width=2pt] (node3) -- (node8); 
  
\node [right= 0.1 cm of node8] (res0) {\pbox{2.5cm}{
$0.004$ \\ 
$0.279$ \\ 
}} ;
\node [draw,circle] (node7) at (3.5, -17.5) {7} ;
 \draw [latex'-, red,line width=2pt] (node8) -- (node7); 
  
\node [below= 0.1 cm of node7] (res0) {\pbox{2.5cm}{
$-0.003$ \\ 
$0.279$ \\ 
}} ;
\node [draw,circle] (node9) at (7.0, -17.5) {9} ;
 \draw [latex'-, red,line width=2pt] (node8) -- (node9); 
  
\node [right= 0.1 cm of node9] (res0) {\pbox{2.5cm}{
$0.003$ \\ 
$0.279$ \\ 
}} ;
\node [draw,circle] (node10) at (7.0, -21.0) {10} ;
 \draw [latex'-, red,line width=2pt] (node9) -- (node10); 
  
\node [right= 0.1 cm of node10] (res0) {\pbox{2.5cm}{
$0.001$ \\ 
$0.279$ \\ 
}} ;
\node [draw,circle] (node11) at (7.0, -24.5) {11} ;
 \draw [latex'-, red,line width=2pt] (node10) -- (node11); 
  
\node [right= 0.1 cm of node11] (res0) {\pbox{2.5cm}{
$-0.0$ \\ 
$0.279$ \\ 
}} ;
\node [draw,circle] (node12) at (3.5, -3.5) {12} ;
 \draw [-latex',green,line width=2pt] (node0) -- (node12); 
  
\node [right= 0.1 cm of node12] (res0) {\pbox{2.5cm}{
$3.719$ \\ 
$0.001$ \\ 
}} ;
\node [draw,circle] (node13) at (3.5, -7.0) {13} ;
 \draw [-latex',green,line width=2pt] (node12) -- (node13); 
  
\node [right= 0.1 cm of node13] (res0) {\pbox{2.5cm}{
$3.758$ \\ 
$0.016$ \\ 
}} ;
\node [draw,circle] (node14) at (3.5, -10.5) {14} ;
 \draw [-latex',green,line width=2pt] (node13) -- (node14); 
  
\node [right= 0.1 cm of node14] (res0) {\pbox{2.5cm}{
$3.781$ \\ 
$0.024$ \\ 
}} ;
\end{tikzpicture}

\end{scriptsize}
 \subcaption{ t=0}
\end{subfigure}
\begin{subfigure}{0.49\columnwidth} \centering
\centering
 \begin{scriptsize}
\begin{tikzpicture}[scale=0.3]
\node [draw,circle] (node0) at (0,0) {0} ;
\node [draw,circle] (node1) at (0.0, -3.5) {1} ;
 \draw [-latex',green,line width=2pt] (node0) -- (node1); 
  
\node [left= 0.1 cm of node1] (res0) {\pbox{2.5cm}{
$1.004$ \\ 
$0.005$ \\ 
}} ;
\node [draw,circle] (node2) at (0.0, -7.0) {2} ;
 \draw [latex'-, red,line width=2pt] (node1) -- (node2); 
  
\node [left= 0.1 cm of node2] (res0) {\pbox{2.5cm}{
$0.98$ \\ 
$0.02$ \\ 
}} ;
\node [draw,circle] (node3) at (0.0, -10.5) {3} ;
 \draw [latex'-, red,line width=2pt] (node2) -- (node3); 
  
\node [left= 0.1 cm of node3] (res0) {\pbox{2.5cm}{
$0.942$ \\ 
$0.041$ \\ 
}} ;
\node [draw,circle] (node4) at (0.0, -14.0) {4} ;
 \draw [-latex',green,line width=2pt] (node3) -- (node4); 
  
\node [left= 0.1 cm of node4] (res0) {\pbox{2.5cm}{
$0.946$ \\ 
$0.042$ \\ 
}} ;
\node [draw,circle] (node5) at (0.0, -17.5) {5} ;
 \draw [-latex',green,line width=2pt] (node4) -- (node5); 
  
\node [left= 0.1 cm of node5] (res0) {\pbox{2.5cm}{
$0.949$ \\ 
$0.043$ \\ 
}} ;
\node [draw,circle] (node6) at (0.0, -21.0) {6} ;
 \draw [-latex',green,line width=2pt] (node5) -- (node6); 
  
\node [left= 0.1 cm of node6] (res0) {\pbox{2.5cm}{
$0.951$ \\ 
$0.044$ \\ 
}} ;
\node [draw,circle] (node8) at (3.5, -14.0) {8} ;
 \draw [latex'-, red, style = dashed ,line width=2pt] (node3) -- (node8); 
  
\node [right= 0.1 cm of node8] (res0) {\pbox{2.5cm}{
$0.003$ \\ 
$0.115$ \\ 
}} ;
\node [draw,circle] (node7) at (3.5, -17.5) {7} ;
 \draw [latex'-, red,line width=2pt] (node8) -- (node7); 
  
\node [below= 0.1 cm of node7] (res0) {\pbox{2.5cm}{
$-0.0$ \\ 
$0.115$ \\ 
}} ;
\node [draw,circle] (node9) at (7.0, -17.5) {9} ;
 \draw [latex'-, red,line width=2pt] (node8) -- (node9); 
  
\node [right= 0.1 cm of node9] (res0) {\pbox{2.5cm}{
$0.002$ \\ 
$0.115$ \\ 
}} ;
\node [draw,circle] (node10) at (7.0, -21.0) {10} ;
 \draw [latex'-, red,line width=2pt] (node9) -- (node10); 
  
\node [right= 0.1 cm of node10] (res0) {\pbox{2.5cm}{
$0.001$ \\ 
$0.115$ \\ 
}} ;
\node [draw,circle] (node11) at (7.0, -24.5) {11} ;
 \draw [latex'-, red,line width=2pt] (node10) -- (node11); 
  
\node [right= 0.1 cm of node11] (res0) {\pbox{2.5cm}{
$-0.0$ \\ 
$0.115$ \\ 
}} ;
\node [draw,circle] (node12) at (3.5, -3.5) {12} ;
 \draw [-latex',green, style = dashed ,line width=2pt] (node0) -- (node12); 
  
\node [right= 0.1 cm of node12] (res0) {\pbox{2.5cm}{
$3.567$ \\ 
$0.732$ \\ 
}} ;
\node [draw,circle] (node13) at (3.5, -7.0) {13} ;
 \draw [-latex',green,line width=2pt] (node12) -- (node13); 
  
\node [right= 0.1 cm of node13] (res0) {\pbox{2.5cm}{
$3.583$ \\ 
$0.737$ \\ 
}} ;
\node [draw,circle] (node14) at (3.5, -10.5) {14} ;
 \draw [-latex',green,line width=2pt] (node13) -- (node14); 
  
\node [right= 0.1 cm of node14] (res0) {\pbox{2.5cm}{
$3.593$ \\ 
$0.741$ \\ 
}} ;
\end{tikzpicture}

\end{scriptsize}
\subcaption{ t=1 }
\end{subfigure}%
 \caption{Directions of  active flows $\bm{f}$  and \acp{DLMP} $(y^p,y^q )$ at the solution given by Algorithm \protect\ref{algorithm:primaldual}. Saturated lines are dashed.}
 \label{fig:resFlows}
 \end{figure}

The solutions show that the active (and reactive) \acp{DLMP} obtained for each time period are close to the \acp{DLMP} at the root node $(\dlmpn0, \dlmqn0)$, with the following exceptions:
\begin{itemize}
\item For the branch composed of nodes $8,7,9,10,11$,  active \acp{DLMP} are close to~$0.0$ 
  due to the presence of renewable production (at null cost) at node 11, and of negative load at node $7$, which together fully compensate for the demand on this branch.
  Since Line $(3,8)$ is saturated, no energy can be exported further.
\item  Active \acp{DLMP} on the branch composed of nodes $(12,13,14)$ at  $t=1$ are much larger than on other nodes: this is explained by the congestion  of line $(0,12)$.
\item  The \ac{DLMP} for node 7 and $t=0$ is \emph{strictly negative}: the (negative) consumption for this node is at its upper bound $p_{7,0}=\overline{P}_{7,0}=-0.173$.  The negative \ac{DLMP} suggests that the system will be better off if less power is {injected} by node 7. 
\end{itemize}


\begin{figure}[!ht]
\centering
\includegraphics[width=1.0\columnwidth]{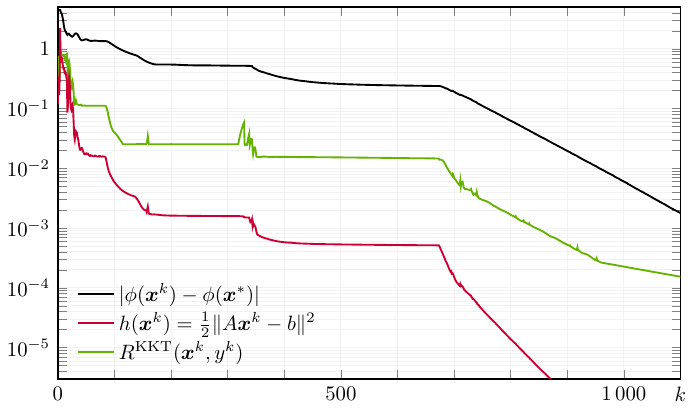}
 \caption{Convergence of last iterate $\xallk{\iter}$%
}
  \label{fig:lastiterate}  
\end{figure} 
\begin{figure}[!ht]
\centering
\includegraphics[width=1.0\columnwidth]{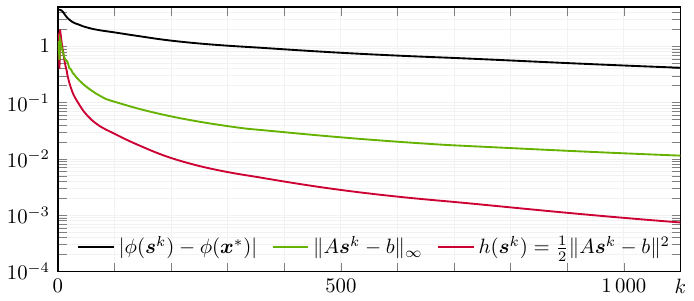}
 \caption{Convergence of ergodic average~$\Sallk{\iter}$} 
 \label{fig:ergodicaverage} 
\end{figure}

Convergence of Algorithm~\ref{algorithm:primaldual} for the 15-bus network is shown in \Cref{fig:lastiterate,fig:ergodicaverage}.
\Cref{fig:lastiterate} displays 
the convergence of the last iterate with respect to various criteria: convergence of~$\cost{\xallk{\iter}}$  to the optimal cost~$\cost{\xallsol}$, convergence to zero of the primal residuals~$\constraint{\xallk{\iter}}$ 
 and convergence to zero of the 
KKT residual
$
 \mixedcriterion{\xallk{\iter}}{\dualk{\iter}}=
 \bigdistt{\infty}{(\partial_{\xall},-\partial_{\dualvariable})\lagrangian{\xallk{\iter}}{\dualk{\iter}}}{0},
$

where $\lagrangian{\xall}{\dualvariable}=\extcost{\xall}+\inner{\dualvariable}{\Aall\xall-\bvector}$ denotes the Lagrangian of~\eqref{genericP},~\cite{luke18}; as well as the convergence of the \acp{DLMP} $\dualk{\iter}$ to stationarity.
\Cref{fig:ergodicaverage} shows the convergence to zero of the primal infeasiblity in the ergodic average~$\Sallk{\iter}$, as predicted by~\Cref{theorem:convergence}. 


\section{Conclusion}
\label{sec:conclusion}
We developed a novel distributed and privacy-preserving algorithm for the computation of distribution locational marginal prices. Our computational strategy builds on extends state-of-the-art block coordinate descent algorithms for convex optimization problems with affine coupling constraints. 
Non-convex versions of \ac{PPDLMP} will be investigated in the future. We also plan to conduct extensions of this work where the electric network is exposed to stochastic uncertainty. 
\appendices
\label{sec:appendix}
%

\section{The relation of Algorithm \ref{algorithm:genericprimaldual} to Algorithm \ref{algorithm:primaldual} }
\label{sec:Algoequivalent}
In this section we show that Algorithm \ref{algorithm:genericprimaldual} contains Algorithm~\ref{algorithm:primaldual} as a special case. In the latter algorithm, the sampling takes values from subsets of the form $\block=\{0,\aggregator\}$, where  $\aggregator\in\{1,\ldots,\cardinality{\Aggregator}\}\equiv\{1,\ldots,\nbaggregators\}$, with probability $\ProbaI{\block}=1/\nbaggregators$ for each~$\block$. Thus, $\dimension=\nbaggregators+1$, $\probai{0}=1$, and $\probai{\coordinate}=1/{\nbaggregators}$ if $\coordinate\in\{1,\ldots,\nbaggregators\}$. The weighting matrix~$\Dimension$ is given by $\Dimension=\diag(\identityI{m_{0}},\nbaggregators\identityI{m_{1}},\ldots,\nbaggregators\identityI{m_{\nbaggregators}})$, where $m_{0}$ is the dimension of the feasible set of the \ac{DSO}, and $m_{a}$ is the dimension of the feasible set of aggregator $a\in\Aggregator$. Now, define the function~$\riskfunction$ in~\eqref{genericP} as $  \risk{\xall} =\riskDSO{\xDSO} + \sum_{\aggregator\in\Aggregator} \riskLA{\aggregator}{\xLA{\aggregator}}
$, where $\riskDSOfunction=\indicatorfunction{\xDSO\in\XDSO}$ and $\riskLAfunction{\aggregator}=\indicatorfunction{\xLA{\aggregator}\in\XLA{\aggregator}}$
, in which $\indicatorfunction{\cdot}$ denotes the indicator function.
If the load aggregator~$\aggregator$ is chosen at step~$\iter$, Line~\ref{algorithm:genericdualk} in Algorithm~\ref{algorithm:genericprimaldual} becomes
\begin{align*}
&\dualk{\iter+1}=\dualk{\iter} +  \dualstepsize \Aall   \Dimension (\xallk{\iter+1} -\xallk{\iter})  + \altdualk{\iter+1}\\
&=\dualk{\iter} +  \dualstepsize \ADSO    (\xDSOk{\iter+1} -\xDSOk{\iter})+  \dualstepsize \nbaggregators \ALA{\aggregator}   
(\xLAk{\aggregator}{\iter+1} -\xLAk{\aggregator}{\iter})  + \altdualk{\iter+1}\\
&=\dualk{\iter} {+}\,  \dualstepsize [\ADSO    (2\xDSOk{\iter+1} {-}\,\xDSOk{\iter}) {-}\,\bvectorDSO] {+}\,  \dualstepsize \nbaggregators \ALA{\aggregator}   
(\xLAk{\aggregator}{\iter{+}1} {-}\,\xLAk{\aggregator}{\iter}) {+}\, \LAaltdualk{\iter{+}1}
\end{align*}
where we define $\LAaltdualk{\iter}=\altdualk{\iter} - \dualstepsize[ \ADSO    \xDSOk{\iter}-\bvectorDSO]$.
Exploiting Line~\ref{algorithm:genericaltdualk} in Algorithm~\ref{algorithm:genericprimaldual}, we find that~$\LAaltdualk{\iter}$ can be computed locally  and inductively by choosing the initial condition $\LAaltdualk{0}=\dualstepsize \sum_{\aggregator\in\Aggregator}(\ALA{\aggregator}\xLAk{\aggregator}{0}  - \bvectorLA{\aggregator}) 
$ initially, then by letting 
\begin{equation*}\label{LAaltdualk}
 \LAaltdualk{\iter+1} =  \LAaltdualk{\iter} + \dualstepsize   \ALA{\aggregator} (\xLAk{\aggregator}{\iter+1}-\xLAk{\aggregator}{\iter}) 
 \equiv 
 \LAaltdualk{\iter} + \dualstepsize  \LAaltaltdualk{\iter} 
 \end{equation*}
 where $\LAaltaltdualk{\iter} =  \ALA{\aggregator} (\xLAk{\aggregator}{\iter+1}-\xLAk{\aggregator}{\iter})$, and Line~\ref{algorithm:genericdualk}  rewrites as
 \begin{equation*}\label{newDSOLAdualk}
\dualk{\iter+1} 
=
\dualk{\iter} +  \dualstepsize [ \ADSO    (2\xDSOk{\iter+1} -\xDSOk{\iter})-\bvectorDSO]+  \dualstepsize (\nbaggregators+1) \LAaltaltdualk{\iter}  + \LAaltdualk{\iter} .
\end{equation*}

\section{Convergence analysis}
 \label{section:convergenceanalysis}
The proof for Theorem~\ref{theorem:convergence} 
uses the reduction of Algorithm~\ref{algorithm:genericprimaldual} to the simpler Algorithm~\ref{algorithm:genericprimalonly}. It is straightforward to show the equivalence between these two schemes if we set~$\thetak{\iter}=1/(\iter+1)$
and $\dualk{\iter} = (\dualstepsize/\thetak{\iter}) (\Aall\Zallk{\iter}-\bvector) $ \cite{malitsky17}.
\begin{algorithmenv}[t]
\caption{Reduction of Algorithm~\ref{algorithm:genericprimaldual} to the form~\cite{tseng08apgm}\label{algorithm:genericprimalonly}}
\removelatexerror
\RestyleAlgo{tworuled}%
\centering
\vspace{-2.5pt}
\begin{algorithm}[H]
\setcounter{AlgoLine}{0}
\DontPrintSemicolon
\SetAlgoNoLine%
\SetKwInOut{Parameters}{{\textbf{Parameters}}}
\SetKwInOut{Init}{{\textbf{Initialization}}}
\SetKwFor{For}{for}{do}{}
\Parameters{$\Dimension$, 
$\dualstepsize>0$, $ \reviseagain{\Stepsizek{\iter}} $, 
$(\thetak{\iter})_{\iter\geq 0}$} 
\Init{$\xallk{0}=\Sallk{0}\in\Xall$} 
\smallskip 
\SetAlgoLined
 \nonl \For{$\iter=0,1,2,\dots$}{
%
%
\routinenl 
$\Zallk{\iter} = (1-\thetak{\iter}) \Sallk{\iter} + \thetak{\iter} \xallk{\iter} $ \; \label{algorithm:genericZallk}
\routinenl
\text{{\textbf{draw block~$\block\in\Block$ at random according to~$\Probas$} }} \;
\routinenl \label{algorithm:genericxIk}
$ \xIk{\block}{\iter+1} = \arg\min_{\altxI{\block}}\big\{ \inner{\grad\costI{\block}{\reviseagain{\xIk{\block}{\iter}}} + \frac\dualstepsize{\thetak{\iter}}\gradI{\block} \constraint{\Zallk{\iter}}}{\altxI{\block}}$\;\nonl
\qquad\qquad\qquad\qquad\  $ + \riskI{\block}{\altxI{\block}} + \frac{1}{2} \normx{\altxI{\block}-\xIk{\block}{\iter}}{\DimensionI{\block}\StepsizekI{\iter}{\block}}^2 \big\}
$ \; 
\routinenl
$\xIk{\minusblock}{\iter+1} = \xIk{\minusblock}{\iter} $ \; \label{algorithm:genericxmIk} 
%
%
\routinenl
$\Sallk{\iter+1}  =   \Zallk{\iter}    +   \thetak{\iter} \Dimension (\xallk{\iter+1} -\xallk{\iter}) $ \; \label{algorithm:genericSallk}
}
 \end{algorithm}
\end{algorithmenv}
For analysis purposes we introduce the auxiliary sequence
\begin{equation}\label{generichatxall}
\begin{array}{l}
\hatxallk{\iter+1} =  \arg\min_{\altxall}\big\{ \inner{\grad\cost{\reviseagain{\xallk{\iter}}} + \frac\dualstepsize{\thetak{\iter}}\grad \constraint{\Zallk{\iter}}}{\altxall}
\qquad\ \\\hfill
+ \risk{\altxall} + \frac{1}{2} \normx{\altxall-\xallk{\iter}}{\Dimension\Stepsizek{\iter}}^2 \big\}.
\end{array}
\end{equation} 
The iterate $\hatxallk{\iter+1}$ corresponds to the next fictitious state if all coordinates were to perform an update in parallel. We now illustrate the main steps involved in proving convergence of the iterates produced by running Algorithm \ref{algorithm:genericprimalonly}.

\subsection{Separable expectations for block coordinate sampling}
For $i\in\{1,\ldots,d\}$ let $U_{i}$ be the $m\times m$ block unitary matrix of the form $U_{i}=\diag(0,\ldots,\identityI{m_{i}},0,\ldots,0)$. Clearly $\sum_{i=1}^{d}U_{i}=\identityI{m}$, and applying the matrix $U_{i}$ to the left of a vector $\dxall=(\dxall_{1},\ldots,\dxall_{d})^{\top}$ gives $U_{i}\dxall=(0,\ldots,\dxall_{i},\ldots,0)^{\top}\in\Real^{m}$. For $\block\in\Block$, define the $m\times m$ matrix $\UnitI{\block}\eqd\sum_{i\in I}U_{i}.$ We have $\expectation{\UnitI{\block}\Dimension}=\identityI{m}$, with $\block \sim \mathcal{U}(\Block)$, and we define $\constraintcovariance = \expectation{\UnitI{\block}\Dimension\Aall^\transpose\Aall\Dimension\UnitI{\block}}$.
It follows from the quadratic form of~$\constraintfunction$ that 
\begin{equation} \label{meanconstraintdeviation}
 \expectation{ \constraint{\xall+\UnitI{\block}\Dimension\dxall} }
 =
  \constraint{\xall} +  \inner{\grad\constraint{\xall}}{\dxall}+ \frac{1}{2}\normx{\dxall}{\constraintcovariance}^2
  .
\end{equation}
\old{
Similarly, since~$\costfunction$ is smooth,  we find form~\eqref{smoothness},
\begin{equation} \label{meancostdeviation}
 \expectation{ \cost{\xall+\UnitI{\block}\Dimension\dxall} }
 \leq
  \cost{\xall} +  \inner{\grad\cost{\xall}}{\dxall}+ \frac{1}{2}\normx{\dxall}{\costcovariance}^2 
  ,
\end{equation}
where $\costcovariance=\expectation{\UnitI{\block}\Dimension\Smoothness\Dimension\UnitI{\block}}$.
}

Let $\filtrationk{\iter}\defeq 
\sigmaalgebrax{\xallk{0},\Sallk{0},\Zallk{0},\dots,\xallk{\iter},\Sallk{\iter},\Zallk{\iter}}
$ denote the history of the process up to step~$\iter$.
We infer the following result for Algorithm~\ref{algorithm:genericprimalonly}, which corresponds to an Expected Separable Overapproximation (ESO), as introduced in \cite{RicTak14,fercoq15,RicTak16}.
\begin{lemma} \label{lem:writehofs}In Algorithm~\ref{algorithm:genericprimalonly},
\begin{align*} 
 \condexpectation{\filtrationk{\iter}}{ \constraint{\Sallk{\iter+1}} }
 & =
  \constraint{\Zallk{\iter}}   \!+\! \thetak{\iter} \inner{\grad\constraint{\Zallk{\iter}}}{\hatxallk{\iter+1}\!\!-\!\xallk{\iter}} \\ 
 & + \frac{\thetak{\iter}^2}{2}\Normx{\hatxallk{\iter+1}\!\!-\!\xallk{\iter}}{\constraintcovariance}^2
\end{align*} 
\end{lemma}
\begin{proof}
Lines~\ref{algorithm:genericxIk} and~\ref{algorithm:genericxmIk} in Algorithm~\ref{algorithm:genericprimalonly} reduce to
$ 
\xallk{\iter+1} = \xallk{\iter} + \UnitI{\block} (\hatxallk{\iter+1}-\xallk{\iter}) 
$. It follows that 
 Line~\ref{algorithm:genericSallk} rewrites as 
\begin{equation} \label{newSallk}
\Sallk{\iter+1}  =   \Zallk{\iter}    +   \thetak{\iter}  \UnitI{\block} \Dimension (\hatxallk{\iter+1}-\xallk{\iter}). 
\end{equation}
\reviseagain{\Cref{lem:writehofs} follows by combining~\eqref{newSallk} with~\eqref{meanconstraintdeviation}.}

\end{proof}
\subsection{Auxiliary facts}
Lemmas~\ref{lemma:descentargumentcostrisk} and~\ref{lemma:descentargumentconstraint} will serve as descent arguments for Algorithm~\ref{algorithm:genericprimalonly}.

 \begin{lemma}[Proximal step] \label{lemma:descentargumentcostrisk}
In Algorithm~\ref{algorithm:genericprimalonly},  $\forall\xall\in\Xall$,
\begin{equation} \label{descentargument}
 \risk{\hatxallk{\iter+1}}+\stronglyconvex{\hatxallk{\iter+1}} \leq \risk{\xall}+\stronglyconvex{\xall} - \frac{1}{2} \normx{\xall-\hatxallk{\iter+1}}{\Dimension\Stepsizek{\iter}}^2 
, 
\end{equation}
where 
\begin{equation}\label{stronglyconvex}
\nocolsep\begin{array}{l}
\stronglyconvex{\xall} = \cost{\reviseagain{\xallk{\iter}}} +  \inner{\grad\cost{\reviseagain{\xallk{\iter}}}}{\xall-\reviseagain{\xallk{\iter}}} 
\qquad \qquad \qquad \\ \hfill
+
\inner{\frac\dualstepsize{\thetak{\iter}}\grad \constraint{\Zallk{\iter}}}{\xall-\Zallk{\iter}}
+ \frac{1}{2} \normx{\xall-\xallk{\iter}}{\Dimension\Stepsizek{\iter}}^2  
.
\end{array}
\end{equation}
\end{lemma}
\begin{proof}
 Equation~\eqref{generichatxall} rewrites as $\hatxallk{\iter+1}= \arg\min_{\altxall}\{\risk{\altxall}+\stronglyconvex{\altxall}\}$.
Hence, $0\in\partial\risk{\hatxallk{\iter+1}}+\grad\stronglyconvex{\hatxallk{\iter+1}}$, and the result follows by strong convexity of $\riskfunction+\stronglyconvexfunction$  with modulus $\Dimension\Stepsizek{\iter}$.
\end{proof}
%
To proceed, observe that Line~\ref{algorithm:genericZallk} in Algorithm~\ref{algorithm:genericprimalonly} rewrites as
\begin{align}
\label{homothecy1}
&\Zallk{\iter} - \Sallk{\iter} = \thetak{\iter} ( \xallk{\iter} - \Sallk{\iter} )    
,
\\
\label{homothecy2}
& (\Zallk{\iter} - \xallk{\iter}) = \tfrac{1-\thetak{\iter}}{\thetak{\iter}} (  \Sallk{\iter}- \Zallk{\iter} )    
.
\end{align}

\begin{lemma}\label{lemma:descentargumentconstraint}
In Algorithm~\ref{algorithm:genericprimalonly}, for any $\xallsol \in \Xallsol$:
\begin{equation} \label{twostarexpectation}
\nocolsep
\begin{array}{l}
 \condexpectation{\filtrationk{\iter}}{ \constraint{\Sallk{\iter+1}} - \constraint{\xallsol}} 
 =
 \thetak{\iter} \inner{\grad\constraint{\Zallk{\iter}}}{\hatxallk{\iter+1}-\xallsol}
 \ \, \\\hfill
+(1-\thetak{\iter})^2 (\constraint{\Sallk{\iter}}  - \constraint{\xallsol}) 
- \thetak{\iter}^2 (\constraint{\xallk{\iter}}-\constraint{\xallsol})  
 \\\hfill
+ \frac{\thetak{\iter}^2}{2}\normx{\hatxallk{\iter+1}-\xallk{\iter}}{\constraintcovariance}^2 .
\end{array}
\end{equation}
\end{lemma}
\begin{proof}

Proceeding as in \cite[(25)-(28)]{luke18}, we write using \eqref{homothecy2}
\begin{equation*}
\inner{\grad\constraint{\Zallk{\iter}}}{\xallsol\!-\!\xallk{\iter}}
 \refereq{\eqref{homothecy2}}{=}  
 \inner{\grad\constraint{\Zallk{\iter}}}{(\xallsol\!-\!\Zallk{\iter})+\tfrac{1\!-\!\thetak{\iter}}{\thetak{\iter}} (    \Sallk{\iter}\!-\!\Zallk{\iter})} 
 \end{equation*}
 we get, using $h(x)=h(y) +  \inner{\grad\constraint{x}}{y} + \tfrac{1}{2}\norm{A(x-y)}^2$ :
 \begin{equation*}
 \begin{split}
   & \constraint{\xallsol}-\constraint{\Zallk{\iter}} -\tfrac{1}{2}\norm{\Aall(\Zallk{\iter}-\xallsol)}^2  \\
 & + \tfrac{1-\thetak{\iter}}{\thetak{\iter}} [\constraint{\Sallk{\iter}}-\constraint{\Zallk{\iter}} -\tfrac{1}{2}\norm{\Aall(\Sallk{\iter}-\Zallk{\iter})}^2]
\end{split}
\end{equation*}
\begin{equation} \label{someresult}
\begin{split}
\refereq{\eqref{homothecy1}}{=}
& \constraint{\xallsol}+\tfrac{1-\thetak{\iter}}{\thetak{\iter}} \constraint{\Sallk{\iter}}-\tfrac{1}{\thetak{\iter}} \constraint{\Zallk{\iter}} 
\\ & -\tfrac{1}{2}\norm{\Aall(\Zallk{\iter}-\xallsol)}^2 -\tfrac{\thetak{\iter}(1-\thetak{\iter})}{2}\norm{\Aall ( \xallk{\iter} - \Sallk{\iter} )}^2
.
\end{split}
\end{equation}
By combining Line~\ref{algorithm:genericZallk} in Algorithm~\ref{algorithm:genericprimalonly} 

$\Aall(\Zallk{\iter}-\xallsol)
= (1-\thetak{\iter}) \Aall(\Sallk{\iter}-\xallsol) + \thetak{\iter} \Aall (\xallk{\iter}-\xallsol)$ and    $2\inner{c-b}{d-b}= \norm{c-b}^2 + \norm{d-b}^2-\norm{d-c}^2, \forall b,c,d$, we get:
\begin{equation}\label{someresultbis}
\nocolsep\begin{array}{l}
\frac{1}{2}\norm{ \Aall(\Zallk{\iter}-\xallsol))}^2
= \frac{1-\thetak{\iter}}{2}\norm{ \Aall(\Sallk{\iter}-\xallsol) }^2 
\\\hfill
+ \frac{\thetak{\iter}}{2}\norm{ \Aall (\xallk{\iter}-\xallsol) }^2
- \frac{\thetak{\iter}(1-\thetak{\iter})}{2} \norm{\Aall(\Sallk{\iter}-\xallk{\iter})}^2
\\\hfill
\quad = 
 (1-\thetak{\iter})(\constraint{\Sallk{\iter}}-\constraint{\xallsol}) + \thetak{\iter}(\constraint{\xallk{\iter}}-\constraint{\xallsol})
 \\\hfill- \frac{\thetak{\iter}(1-\thetak{\iter})}{2} \norm{\Aall(\Sallk{\iter}-\xallk{\iter})}^2
 ,
\end{array} 
\end{equation}
where we have used $\grad\constraint{\xallsol}=0$. By subtracting \eqref{someresultbis} to ~\eqref{someresult}, we find
\begin{equation}\label{someresultend}
\begin{array}{l}
\inner{\grad\constraint{\Zallk{\iter}}}{\xallsol-\xallk{\iter}}
=
\frac{(1-\thetak{\iter})^2}{\thetak{\iter}} (\constraint{\Sallk{\iter}}-\constraint{\xallsol})
\quad\\\hfill 
-\frac{1}{\thetak{\iter}} (\constraint{\Zallk{\iter}} -\constraint{\xallsol})  - \thetak{\iter}(\constraint{\xallk{\iter}}-\constraint{\xallsol})
.
\end{array} 
\end{equation}
The result follows by combining~\Cref{lem:writehofs} with~\eqref{someresultend}.
\end{proof}
 
\subsection{Extrapolation}
The next results characterize the sequence~$(\Sallk{\iter})$ as a linear combination of the past primal iterates. This characterization is a generalization of \cite[Lemma~2]{fercoq15}, and its proof is similar to that work. 

\begin{lemma}\label{lemma:extrapolation}
 In Algorithm~\ref{algorithm:genericprimalonly}, we have
\begin{equation}
\label{Sallk}
 \Sallk{\iter} = \sum_{\altiter=0}^{\iter} \gammacoefkl{\iter}{\altiter} \xallk{\altiter}
 , \qquad \iter\geq 1,
\end{equation}
where~$(\gammacoefkl{\iter}{\altiter})$ is a collection of diagonal matrices defined by $\gammacoefkl{1}{0}=
\identity-\thetak{0}\Dimension
$, $\gammacoefkl{1}{1}=\thetak{0}\Dimension$, and, for $\iter\geq 1$,
\begin{equation} \label{gammacoefkl}
 \nocolsep
 \gammacoefkl{\iter+1}{\altiter} \,{=} \left\{\!
 \begin{array}{ll}
 (1-\thetak{\iter})\gammacoefkl{\iter}{\altiter}
 &\text{ for } \altiter\,{=}\, 0,\dots,\iter{-}1,
 \\
 (1-\thetak{\iter})\thetak{\iter-1}\Dimension -\thetak{\iter}(\Dimension-\identity)
 &\text{ if } \altiter\,{=}\,\iter,
 \\
 \thetak{\iter}\Dimension
 &\text{ if } \altiter\,{=}\,\iter{+}1.
 \end{array}\right.
\end{equation} 
Besides,
$
  \gammacoefkl{\iter+1}{\iter} =  (1-\thetak{\iter})\gammacoefkl{\iter}{\iter} -\thetak{\iter}(\Dimension-\identity) .
$
\end{lemma}
\begin{proof}
  We proceed by induction.  By combining Lines~\ref{algorithm:genericZallk} and~\ref{algorithm:genericSallk} in Algorithm~\ref{algorithm:genericprimalonly}, we find
\begin{equation}\label{developSallk}
 \Sallk{\iter+1}  
 =  
(1-\thetak{\iter}) \Sallk{\iter} + \thetak{\iter} \xallk{\iter}    +  \thetak{\iter} \Dimension (\xallk{\iter+1} -\xallk{\iter})
\end{equation}
 which yields 
 $  \Sallk{1}  =
(\identity-\thetak{0}\Dimension) \xallk{0} +\thetak{0}\Dimension \xallk{1}$, and the values of $\gammacoefkl{1}{0}$ and~$\gammacoefkl{1}{1}$. Suppose now that~\eqref{Sallk} holds for $\iter\geq 1$, then it follows from~\eqref{developSallk} that
\begin{equation}\label{inspection}
 \nocolsep
 \begin{array}{rcl}
 \Sallk{\iter+1}  
&=&
(1-\thetak{\iter}) \sum_{\altiter=0}^{\iter} \gammacoefkl{\iter}{\altiter} \xallk{\altiter} +\thetak{\iter} \Dimension \xallk{\iter+1} -\thetak{\iter}(\Dimension-\identity)\xallk{\iter}
\\
&=&
 \sum_{\altiter=0}^{\iter-1} (1-\thetak{\iter}) \gammacoefkl{\iter}{\altiter} \xallk{\altiter} +  [ (1-\thetak{\iter}) \gammacoefkl{\iter}{\iter} -\thetak{\iter}(\Dimension-\identity)]\xallk{\iter}\\&&\hfill + \thetak{\iter} \Dimension \xallk{\iter+1}.
\end{array}
\end{equation}  
The lemma follows by inspection of~\eqref{Sallk} and~\eqref{inspection}.
\end{proof}
\smallskip
Now, define $\vextcostfunction=(\extcostIfunction{1},\dots,\extcostIfunction{\dimension})$ and $\hatextcostk{\iter}=\bm{1}^\transpose\hatvextcostk{\iter}$, where
\reviseagain{
\begin{equation} 
 \label{hatvextcostk}
 \hatvextcostk{\iter} = \sum_{\altiter=0}^{\iter} \gammacoefkl{\iter}{\altiter} \vextcost{\xallk{\altiter}}
 , \qquad \iter\geq 1.
\end{equation}
}
By convexity, it follows from~\eqref{Sallk} and~\eqref{hatvextcostk} that  $\hatvextcostk{\iter}\geq\vextcost{\Sallk{\iter}}$ and  $\hatextcostk{\iter}\geq\extcost{\Sallk{\iter}}$.
\begin{lemma} In  Algorithm~\ref{algorithm:genericprimalonly},
\begin{align}
& 
  \condexpectation{\filtrationk{\iter}}{\normx{\xallk{\iter+1}\!\!\!-\!\xallsol}{_{\Dimension^2\Stepsizek{\iter}}}^2} 
=
\normx{\hatxallk{\iter+1}\!\!\!-\!\xallsol}{_{\Dimension\Stepsizek{\iter}}}^2  
\!+\! \normx{\xallk{\iter}\!\!-\!\xallsol}{_{\Dimension(\Dimension-\identity)\Stepsizek{\iter}}}^2
\label{expectationdistance}
\\
&\label{expectationhatextcost}
\reviseagain{
\condexpectation{\filtrationk{\iter}}{\hatextcostk{\iter+1}} 
=
(1-\thetak{\iter})\hatextcostk{\iter} +\thetak{\iter} \extcost{\hatxallk{\iter+1}}
}
,
\end{align}
where~$\hatxallk{\iter+1}$ and~$\hatextcostk{\iter}$ are defined as in~\eqref{generichatxall} and~\eqref{hatvextcostk}. 
\end{lemma}
\begin{proof}
For $\coordinate=1,\dots,\dimension$, some simple algebra gives
\begin{equation*} \! \! \!\!
\condexpectation{\filtrationk{\iter}}{\normx{\xIk{\coordinate}{\iter+1}\!\!-\!\xIsol{\coordinate}}{_{\DimensionI{\block}^2\Stepsizeki{\iter}{\coordinate}}}^2\!} =\normx{\hatxIk{\coordinate}{\iter+1}\!\!-\!\xIsol{\coordinate}}{_{\DimensionI{\coordinate}\Stepsizeki{\iter}{\coordinate}}}^2\!\!+ (\probai{\coordinate}^{-1}\!-1)
  \normx{\xIk{\coordinate}{\iter}\!\!-\!\xIsol{\coordinate}}{_{\DimensionI{\coordinate}\Stepsizeki{\iter}{\coordinate}}}^2
\end{equation*}
Summing up the above for $\coordinate=1,\dots,\dimension$ gives~\eqref{expectationdistance}.
Next, observe that
$
\condexpectation{\filtrationk{\iter}}{\extcostI{\coordinate}{\xallk{\iter+1}}}
 =
 \probai{\coordinate} \extcostI{\coordinate}{\hatxallk{\iter+1}} +(1-\probai{\coordinate}) \extcostI{\coordinate}{\xallk{\iter}} 
 $ for $\coordinate\in\{1,\dots,\dimension\}$,
 which in matrix form rewrites as
\begin{equation} \label{expectationvextcost}
 \condexpectation{\filtrationk{\iter}}{\vextcost{\xallk{\iter+1}}}
 =
 \Dimension^{-1} \vextcost{\hatxallk{\iter+1}} +(\identity- \Dimension^{-1}) \vextcost{\xallk{\iter}} .
\end{equation}
It follows that
\begin{equation*}
 \nocolsep
 \begin{array}{ll}
  \condexpectation{\filtrationk{\iter}}{\hatvextcostk{\iter+1}}
  \hspace{-15mm}&\hspace{15mm} 
  \refereq{\eqref{hatvextcostk}}{=}
  \sum_{\altiter=0}^{\iter-1} \gammacoefkl{\iter+1}{\altiter} \vextcost{\xallk{\altiter}}+ \gammacoefkl{\iter+1}{\iter} \vextcost{\xallk{\iter}}
  \\&\hfill
  +\gammacoefkl{\iter+1}{\iter+1} \condexpectation{\filtrationk{\iter}}{\vextcost{\xallk{\iter+1}}}
  \\ &\refereq{\eqref{gammacoefkl}}{=}
  (1-\thetak{\iter})\sum_{\altiter=0}^{\iter-1} \gammacoefkl{\iter}{\altiter} \vextcost{\xallk{\altiter}}+ \gammacoefkl{\iter+1}{\iter} \vextcost{\xallk{\iter}}
  \\&\hfill
  + \thetak{\iter}\Dimension \condexpectation{\filtrationk{\iter}}{\vextcost{\xallk{\iter+1}}}
  \\ &\refereq{\eqref{expectationvextcost}}{=}
  (1-\thetak{\iter}) \sum_{\altiter=0}^{\iter-1} \gammacoefkl{\iter}{\altiter} \vextcost{\xallk{\altiter}}
    \\&\hfill
+ [\gammacoefkl{\iter+1}{\iter} + \thetak{\iter}   (\Dimension-\identity)] \vextcost{\xallk{\iter}} +\thetak{\iter} \vextcost{\hatxallk{\iter+1}}
  \\ 
&\refereq{\eqref{hatvextcostk}}{=}
 (1-\thetak{\iter})\hatvextcostk{\iter} +\thetak{\iter} \vextcost{\hatxallk{\iter+1}}
 ,
 \end{array}
\end{equation*}
which yields~\eqref{expectationhatextcost} since $\extcostfunction=\bm{1}^\transpose\vextcostfunction$ and $\hatextcostk{\iter}=\bm{1}^\transpose\hatvextcostk{\iter}$.
\end{proof}
 

\subsection{Main descent argument.}
Since $\hatxallk{\iter+1}\in\domain{\riskfunction}$, using \eqref{smoothness} and recognizing $\stronglyconvex{\hatxallk{\iter+1}} $:
\begin{equation*}
\nocolsep
\begin{array}{ll}
  \cost{\hatxallk{\iter+1}} 
 &\refereq{\eqref{smoothness}}{\leq}
\cost{\xallk{\iter}} \!+\!   \inner{\grad\cost{\xallk{\iter}}}{\hatxallk{\iter+1}\!\!-\!\xallk{\iter}}\!+\! \frac{1}{2}\normx{\hatxallk{\iter+1}\!\!-\!\xallk{\iter}}{\Smoothness}^2
\\&
\!\!\!\hspace{-20pt}\refereq{\eqref{stronglyconvex}}{=}
\stronglyconvex{\hatxallk{\iter+1}} \! -\!  \frac\dualstepsize{\thetak{\iter}} \inner{\grad \constraint{\Zallk{\iter}}}{\hatxallk{\iter+1}\!\!-\!\Zallk{\iter}} \!-\! \frac{1}{2} \normx{\hatxallk{\iter+1}\!\!-\!\xallk{\iter}}{_{\Dimension\Stepsize-\Smoothness}}^2  
\
\end{array}
\end{equation*}
We then use \eqref{descentargument} to have  $$\stronglyconvex{\hatxallk{\iter+1}} \leq \risk{\xallsol} -\risk{\hatxallk{\iter+1}}  +\stronglyconvex{\xallsol} - \tfrac{1}{2} \normx{\hatxallk{\iter+1}-\xallsol}{\Dimension\Stepsize}^2,$$
then replace $\stronglyconvex{\xallsol}$ by its expression and use convexity of $\costfunction$ to get:
\begin{equation*}
\nocolsep
\begin{split}
  & \cost{\hatxallk{\iter+1}} 
\leq 
-[ \risk{\hatxallk{\iter+1}} - \risk{\xallsol}  + \tfrac{1}{2} \normx{\hatxallk{\iter+1}-\xallsol}{\Dimension\Stepsize}^2 ] + \cost{\xallsol}\\
& \!-\!  \tfrac\dualstepsize{\thetak{\iter}} \inner{\grad \constraint{\Zallk{\iter}}}{\hatxallk{\iter+1}\!\!-\!\xallsol} \! -\! \tfrac{1}{2} \normx{\hatxallk{\iter+1}\!\!-\!\xallk{\iter}}{\Dimension\Stepsize-\Smoothness}^2  
 \!+\! \tfrac{1}{2} \normx{\xallk{\iter}\!-\!\xallsol}{\Dimension\Stepsize}^2
.
\end{split}
\end{equation*}
 By adding~\eqref{twostarexpectation} multiplied by ${\dualstepsize}/{\thetak{\iter}^2}$ to the above, we find
\begin{equation*}
\begin{split}
&\cost{\hatxallk{\iter{+}1}} 
\,{\leq}\,
{-} \big[  \risk{\hatxallk{\iter{+}1}} {-}\, \risk{\xallsol} \,{+}\,\tfrac{\dualstepsize}{\thetak{\iter}^2} \condexpectation{\filtrationk{\iter}}{ \constraint{\Sallk{\iter{+}1}} {-}\, \constraint{\xallsol}} \\
& \, 
+ \tfrac{1}{2} \normx{\hatxallk{\iter+1}\!\!-\!\xallsol}{\Dimension\Stepsize}^2 \big]
+  \tfrac{\dualstepsize(1-\thetak{\iter})^2}{\thetak{\iter}^2}  (\constraint{\Sallk{\iter}} - \constraint{\xallsol}) +\cost{\xallsol} 
  \\
  \,
  &\!-\! \dualstepsize (\constraint{\xallk{\iter}}\!-\!\constraint{\xallsol})  
\!-\! \tfrac{1}{2} \normx{\hatxallk{\iter+1}\!\!-\!\xallk{\iter}}{\Dimension\Stepsize-\Smoothness-\dualstepsize\constraintcovariance}^2  
\!+\! \tfrac{1}{2} \normx{\xallk{\iter}\!-\!\xallsol}{\Dimension\Stepsize}^2 .
\end{split}
\end{equation*}
After adding to the above~\eqref{expectationdistance} multiplied by~$1/2$ and~\eqref{expectationhatextcost} multiplied by~$1/\thetak{\iter}$, and rearranging the terms, we find
\begin{equation*}
\begin{split}
 & \condexpectation{\filtrationk{\iter}}{\tfrac{1}{\thetak{\iter}}[\hatextcostk{^{\iter+1}}\!\!-\!\extcost{\xallsol}] \!+\!\tfrac{\dualstepsize}{\thetak{\iter}^2}[\constraint{\Sallk{\iter+1}} \!-\! \constraint{\xallsol}]}  \!+\! \tfrac{1}{2} \normx{\xallk{\iter+1}\!\!\!-\!\xallsol}{\Dimension^2\Stepsize}^2 
\\& \, \,\leq
\tfrac{1-\thetak{\iter}}{\thetak{\iter}}[\hatextcostk{\iter} - \extcost{\xallsol}] 
 + 
 \tfrac{\dualstepsize(1-\thetak{\iter})^2}{\thetak{\iter}^2}  (\constraint{\Sallk{\iter}} - \constraint{\xallsol})
\\ & \,\,
+ \tfrac{1}{2} \normx{\xallk{\iter}-\xallsol}{\Dimension^2\Stepsize}^2 
  - \dualstepsize (\constraint{\xallk{\iter}}-\constraint{\xallsol})  
- \tfrac{1}{2} \normx{\hatxallk{\iter+1}-\xallk{\iter}}{\Scalek{\iter}}^2  
 ,
\end{split}
\end{equation*}
where we write $\Scalek{\iter}\defeq \Dimension\Stepsizek{\iter}-\Smoothness-\dualstepsize\constraintcovariance$. 
%
%
A Lyapunov function shows in the inequality above on condition that 
\begin{equation}\label{thetacondition}
  ({1-\thetak{\iter+1}})/{\thetak{\iter+1}} \leq  {1}/\thetak{\iter}  \old{\quad\text{and}\quad \Stepsizek{\iter+1}\preceq\Stepsizek{\iter}}.
\end{equation}
Then we can define the Lyapunov function $\textstyle \lyapunovk{\iter} \defeq $
\reviseagain{
\begin{equation*}
\tfrac{1-\thetak{\iter}}{\thetak{\iter}}[\hatextcostk{\iter}  -\extcost{\xallsol} ]
+  \tfrac{\dualstepsize(1-\thetak{\iter})^2 }{\thetak{\iter}^2}(\constraint{\Sallk{\iter}} - \constraint{\xallsol}) 
+  \tfrac{ 1}{2}  \normx{\xallk{\iter}-\xallsol}{\Dimension^2\Stepsizek{\iter}}^2  ,
\end{equation*} 
}%
and the Lyapunov inequality above yields
\begin{equation} \label{argumentsequelthree}
 \condexpectation{\filtrationk{\iter}}{ \lyapunovk{\iter+1}}   
\leq
\lyapunovk{\iter} - \dualstepsize (\constraint{\xallk{\iter}}-\constraint{\xallsol})  - \frac{1}{2}  \normx{\hatxallk{\iter+1}-\xallk{\iter}}{\Scalek{\iter}}^2
.
\end{equation}

The tightest choice for sequence~$(\thetak{\iter})$ allowed by~\eqref{thetacondition} is 
\begin{equation}\label{tightthetasequence} 
\thetak{\iter}=1/(\iter+1) 
, \qquad \forall\iter\geq 0,
\end{equation}
 in which case Algorithm~\ref{algorithm:genericprimalonly} rewrites in primal-dual form as an inexact, block-coordinate variant of the algorithm proposed in~\cite{luke18} with proximal gradient steps for the smooth part of the composite objective, and $ \lyapunovk{\iter}$ reduces to
\begin{align*}
 \lyapunovk{\iter} \,{=}\,  \iter[\hatextcostk{\iter}   {-}\extcost{\xallsol} ]
{+}  \dualstepsize\iter^2(\constraint{\Sallk{\iter}} {-} \constraint{\xallsol}) {+}  \frac{1}{2}  \normx{\xallk{\iter}{-}\xallsol}{\Dimension^2\Stepsizek{\iter}}^2  .
\end{align*}

We are now in a position to show \Cref{theorem:convergence}.
\endedit{}%
\def\proof{\noindent\hspace{2em}{\itshape Proof of \Cref{theorem:convergence}: }}%

\begin{proof}
Recall that Algorithms~\ref{algorithm:genericprimaldual} and~\ref{algorithm:genericprimalonly} are equivalent if $\thetak{\iter}=1/(\iter+1)$
 and $\dualk{\iter} = (\dualstepsize/\thetak{\iter}) (\Aall\Zallk{\iter}-\bvector) $. It is then straightforward to see that~\eqref{earlydiagonalscaling} ensures that $\Scalek{\iter}\succ 0$ in~\eqref{argumentsequelthree}.
 Besides, \eqref{argumentsequelthree} is analogous to~\cite[(39)]{luke18}, with an additional smooth term $\costfunction$.
The rest of the proof relies on arguments due to~\cite{sodolov07,malitsky17}\textemdash
since it follows the exact lines~\cite[pp. 13-15]{luke18}, it will be omitted for brevity. 
%
\end{proof}



\bibliographystyle{IEEEtran}
\bibliography{IEEEabrv,obPDLO-LA}

\begin{thebibliography}{10}
\providecommand{\url}[1]{#1}
\csname url@samestyle\endcsname
\providecommand{\newblock}{\relax}
\providecommand{\bibinfo}[2]{#2}
\providecommand{\BIBentrySTDinterwordspacing}{\spaceskip=0pt\relax}
\providecommand{\BIBentryALTinterwordstretchfactor}{4}
\providecommand{\BIBentryALTinterwordspacing}{\spaceskip=\fontdimen2\font plus
\BIBentryALTinterwordstretchfactor\fontdimen3\font minus
  \fontdimen4\font\relax}
\providecommand{\BIBforeignlanguage}[2]{{%
\expandafter\ifx\csname l@#1\endcsname\relax
\typeout{** WARNING: IEEEtran.bst: No hyphenation pattern has been}%
\typeout{** loaded for the language `#1'. Using the pattern for}%
\typeout{** the default language instead.}%
\else
\language=\csname l@#1\endcsname
\fi
#2}}
\providecommand{\BIBdecl}{\relax}
\BIBdecl

\bibitem{NAPS16}
{National Academies of Sciences, Engineering, and Medicine}, \emph{Analytic
  Research Foundations for the Next-Generation Electric Grid}.\hskip 1em plus
  0.5em minus 0.4em\relax Washington, DC: The National Academies Press, 2016.

\bibitem{SotVig06}
P.~M. {Sotkiewicz} and J.~M. {Vignolo}, ``Nodal pricing for distribution
  networks: efficient pricing for efficiency enhancing dg,'' \emph{IEEE
  Transactions on Power Systems}, vol.~21, no.~2, pp. 1013--1014, 2006.

\bibitem{Hey12}
G.~T. {Heydt}, B.~H. {Chowdhury}, M.~L. {Crow}, D.~{Haughton}, B.~D. {Kiefer},
  F.~{Meng}, and B.~R. {Sathyanarayana}, ``Pricing and control in the next
  generation power distribution system,'' \emph{IEEE Transactions on Smart
  Grid}, vol.~3, no.~2, pp. 907--914, 2012.

\bibitem{Con12}
N.~O'Connell, Q.~Wu, J.~{\O}stergaard, A.~H. Nielsen, S.~T. Cha, and Y.~Ding,
  ``Day-ahead tariffs for the alleviation of distribution grid congestion from
  electric vehicles,'' \emph{Electric Power Systems Research}, vol.~92, pp.
  106--114, 2012.

\bibitem{SinGos10}
R.~Singh and S.~Goswami, ``Optimum allocation of distributed generations based
  on nodal pricing for profit, loss reduction, and voltage improvement
  including voltage rise issue,'' \emph{International Journal of Electrical
  Power \& Energy Systems}, vol.~32, no.~6, pp. 637--644, 2010.

\bibitem{malitsky17}
Y.~Malitsky, ``Chambolle-pock and tseng's methods: relationship and extension
  to the bilevel optimization,'' 06 2017.

\bibitem{luke18}
D.~R. Luke and Y.~Malitsky, \emph{Block-Coordinate Primal-Dual Method for
  Nonsmooth Minimization over Linear Constraints}.\hskip 1em plus 0.5em minus
  0.4em\relax Cham: Springer International Publishing, 2018, pp. 121--147.

\bibitem{malitsky2019primaldual}
Y.~Malitsky, ``The primal-dual hybrid gradient method reduces to a primal
  method for linearly constrained optimization problems,'' 2019.

\bibitem{Li:2014aa}
R.~Li, Q.~Wu, and S.~S. Oren, ``Distribution locational marginal pricing for
  optimal electric vehicle charging management,'' \emph{IEEE Transactions on
  Power Systems}, vol.~29, no.~1, pp. 203--211, 2014.

\bibitem{Hua2014}
S.~Huang, Q.~Wu, S.~S. Oren, R.~Li, and Z.~Liu, ``Distribution locational
  marginal pricing through quadratic programming for congestion management in
  distribution networks,'' \emph{IEEE Transactions on Power Systems}, vol.~30,
  no.~4, pp. 2170--2178, 2014.

\bibitem{farivar13}
M.~{Farivar} and S.~H. {Low}, ``Branch flow model: Relaxations and
  convexification---part i,'' \emph{IEEE Transactions on Power Systems},
  vol.~28, no.~3, pp. 2554--2564, 2013.

\bibitem{Kocuk:2016aa}
B.~Kocuk, S.~S. Dey, and X.~A. Sun, ``Strong socp relaxations for the optimal
  power flow problem,'' \emph{Operations Research}, vol.~64, no.~6, pp.
  1177--1196, 2021/02/20 2016.

\bibitem{papavasiliou18}
A.~{Papavasiliou}, ``Analysis of distribution locational marginal prices,''
  \emph{IEEE Transactions on Smart Grid}, vol.~9, no.~5, pp. 4872--4882, 2018.

\bibitem{Alsaleh:2018aa}
I.~Alsaleh and L.~Fan, ``Distribution locational marginal pricing (dlmp) for
  multiphase systems,'' in \emph{2018 North American Power Symposium (NAPS)},
  2018, pp. 1--6.

\bibitem{Bai:2018aa}
L.~Bai, J.~Wang, C.~Wang, C.~Chen, and F.~Li, ``Distribution locational
  marginal pricing (dlmp) for congestion management and voltage support,''
  \emph{IEEE Transactions on Power Systems}, vol.~33, no.~4, pp. 4061--4073,
  2018.

\bibitem{zhou2017incentive}
X.~Zhou, E.~Dall’Anese, L.~Chen, and A.~Simonetto, ``An incentive-based
  online optimization framework for distribution grids,'' \emph{IEEE
  transactions on Automatic Control}, vol.~63, no.~7, pp. 2019--2031, 2017.

\bibitem{peng18}
Q.~{Peng} and S.~H. {Low}, ``Distributed optimal power flow algorithm for
  radial networks, i: Balanced single phase case,'' \emph{IEEE Transactions on
  Smart Grid}, vol.~9, no.~1, pp. 111--121, 2018.

\bibitem{molzahn2019survey}
D.~K. Molzahn, I.~A. Hiskens \emph{et~al.}, ``A survey of relaxations and
  approximations of the power flow equations,'' \emph{Foundations and
  Trends{\textregistered} in Electric Energy Systems}, vol.~4, no. 1-2, pp.
  1--221, 2019.

\bibitem{Nes18}
Y.~Nesterov, \emph{Lectures on Convex Optimization}, ser. Springer Optimization
  and Its Applications.\hskip 1em plus 0.5em minus 0.4em\relax Springer
  International Publishing, 2018, vol. 137.

\bibitem{chambolle11}
A.~Chambolle and T.~Pock, ``A first-order primal-dual algorithm for convex
  problems with applications to imaging,'' \emph{Journal of Mathematical
  Imaging and Vision}, vol.~40, no.~1, pp. 120--145, May 2011.

\bibitem{tseng08apgm}
P.~Tseng, ``On accelerated proximal gradient methods for convex-concave
  optimization,'' 2008.

\bibitem{jacquot20}
P.~Jacquot, ``{DLMP-based Coordination Procedure for Decentralized Demand
  Response under Distribution Network Constraints},'' Aug. 2020, working paper
  or preprint.

\bibitem{RicTak14}
P.~Richt{\'a}rik and M.~Tak{\'a}{\v c}, ``Iteration complexity of randomized
  block-coordinate descent methods for minimizing a composite function,''
  \emph{Mathematical Programming}, vol. 144, no. 1 
  2014.

\bibitem{fercoq15}
O.~Fercoq and P.~Richt{\'a}rik, ``Accelerated, parallel, and proximal
  coordinate descent,'' \emph{SIAM Journal on Optimization}, vol.~25, no.~4,
  pp. 1997--2023, 2015.

\bibitem{RicTak16}
P.~Richt{\'a}rik and M.~Tak{\'a}{\v c}, ``Parallel coordinate descent methods
  for big data optimization,'' \emph{Mathematical Programming}, vol. 156, no.
  1436-4646, pp. 433--484, 2016.

\bibitem{sodolov07}
M.~Solodov, ``An explicit descent method for bilevel convex optimization,''
  \emph{Journal of Convex Analysis}, vol.~14, pp. 227--238, 04 2007.

\end{thebibliography}

\end{document}